\documentclass[aps,pra, two column,10pt,superscriptaddress,longbibliography,nofootinbib]{revtex4-2}
\usepackage[utf8]{inputenc}
\usepackage{etex}
\usepackage{physics}
\usepackage{amsmath,amssymb,amsthm,amstext,mathrsfs}
\usepackage[colorlinks=true,citecolor=blue,urlcolor=blue]{hyperref}
\usepackage{comment}
\usepackage{eufrak}
\usepackage[pdftex]{graphicx}
\usepackage{enumerate}
\usepackage{times,txfonts}
\usepackage{braket}
\usepackage{svg}
\usepackage{color}
\usepackage{natbib}
\usepackage{amsmath,blkarray}
\usepackage{mathtools}
\usepackage{latexsym}
\usepackage{tabularx, booktabs}
\usepackage{graphics,epstopdf}
\usepackage{graphicx}
\usepackage{float}
\usepackage{amsfonts}
\usepackage{subcaption}
\usepackage{enumerate}
\usepackage{color,soul}

\newcommand{\be}{\begin{equation}}
	\newcommand{\ee}{\end{equation}}
\newcommand{\ba}{\begin{eqnarray}}
	\newcommand{\ea}{\end{eqnarray}}

\newtheorem{thm}{Theorem}

\DeclareMathOperator*{\Motimes}{\text{\raisebox{0.25ex}{\scalebox{0.65}{$\bigotimes$}}}}

\begin{document}

\title{Self-testing in  a constrained prepare-measure scenario sans assuming quantum dimension}

\author{Ritesh K. Singh}
\email{riteshcsm1@gmail.com}
\affiliation{Department of Physics, Indian Institute of Technology Hyderabad, Kandi, Sangareddy, Telangana 502285, India}

\author{Souradeep Sasmal}
\email{souradeep.007@gmail.com}
\affiliation{Institute of Fundamental and Frontier Sciences, University of Electronic Science and Technology of China, Chengdu 611731, China}

\author{S. Nautiyal}
\email{sameernautiyal02@gmail.com}
\affiliation{Department of Physics, Indian Institute of Technology Hyderabad, Kandi, Sangareddy, Telangana 502285, India}

\author{A. K. Pan }
\email{akp@phy.iith.ac.in}
\affiliation{Department of Physics, Indian Institute of Technology Hyderabad, Kandi, Sangareddy, Telangana 502285, India}


\begin{abstract}
We present a device-independent (DI) self-testing protocol in a constrained prepare-measure scenario, based on the $n-$bit parity-oblivious multiplexing (POM) task. In this scenario, a parity-oblivious constraint is imposed on the preparations, allowing us to define a classical bound derived from a preparation noncontextual ontological model. We derive the optimal quantum success probability in the POM task devoid of assuming the dimension of the quantum system, an essential step towards DI self-testing, which has hitherto not been demonstrated in prepare-measure scenario. We demonstrate that the optimal quantum value exceeds preparation noncontextual bound and, as a result, this establishes DI self-testing of the preparations and the measurement devices. Furthermore, by explicitly constructing the required unitaries, we show that the optimal preparations and measurements in an unknown but finite dimensional Hilbert space, responsible for the observed input-output correlations, can be mapped, via an unitary, onto a known finite-dimensional quantum system. Our results thus pave the way for scalable, single system based DI certification protocols in the prepare-measure scenario.
 \end{abstract}
 

\maketitle


\section{Introduction} 
  
Since its inception, Bell’s theorem, originally formulated to address the foundational questions about the viability of a local realist description of quantum theory \cite{Bell1964}, has evolved into a powerful framework for processing quantum information \cite{Brunner2014}.  One of its key strengths lies in its ability to power DI self-testing \cite{Mayers2004, Supic2020} - the strongest form of certification of quantum devices solely from observed input-output statistics, requiring no assumptions about the internal mechanism of the devices or the dimensionality of the underlying quantum systems. This makes Bell’s theorem central to numerous applications such as, quantum key distribution \cite{Ekert1991, Acin2007}, and certified randomness generation \cite{Liu2021, Wooltron2022, Zhang2025}.

However, the practical implementation of DI Bell test poses a significant challenge. In particular, it is required to ensure that the local events of spatially separated parties are space-like separated for a loophole-free Bell test, which remains experimentally demanding. Although progress has been made in performing such tests \cite{Hensen2015, Bierhorst2018, Li2018, Liu2021, Storz2023, Zhao2024}, the design of scalable and robust near-term quantum technologies continues to face considerable constraints. In response to these challenges, prepare-measure certification protocols have emerged as promising alternatives \cite{Brakerski2018, Um2020, Pivoluska2021, Gois2021, Gitton2022, Divianszky2023, Navascues2023, Svegborn2025, Ding2025}. These semi-DI protocols operate under partial assumptions, most commonly, a known upper bound on the Hilbert space dimension of the quantum system, and thus relax the stringent requirements of fully DI protocols. Importantly, such protocols \cite{Galvao2001,Ambainis2008} do not necessitate entanglement, but devices remain uncharacterised. Due to their relative experimental simplicity, semi-DI protocols are gaining considerable relevance in implementations of near-term quantum technologies \cite{Brakerski2018, Um2020, Pivoluska2021, Galvao2001, Ambainis2008, Pawlowski2011, Brunner2013, Tavakoli2015, Guerin2016, Chailloux2016, Himbeeck2017, Tavakoli2018, Mohan2019,Ambainis2019, Saha2019, Saha2019a, Miklin2020,Tavakoli2020, Mukherjee2021, Wei2021, Abhyoudai2023}. 

In a prepare-measure scenario, the distinction between classical and quantum correlations is usually demonstrated via communication tasks involving two parties: a sender (Alice) and a receiver (Bob). Alice encodes classical information into physical systems and transmits them to Bob, who attempts to retrieve the information by performing suitable measurements adhering to the winning rule of the game. The success probability of such tasks serves as a figure of merit. In recent years, a multitude of communication games have been developed, demonstrating a quantum advantage under dimension constraints \cite{Ambainis2008, Galvao2001, Pawlowski2011, Brunner2013, Tavakoli2015, Guerin2016, Chailloux2016, Himbeeck2017, Tavakoli2018, Mohan2019, Saha2019a, Tavakoli2020, Pauwels2022}. These games have found diverse applications in semi-DI self-testing protocols and in the development of quantum technologies, including key distribution \cite{Bennet1984,Cerf2002}, randomness generation \cite{Um2020, Pivoluska2021, Lunghi2015}, and certification of unsharp measurements \cite{Mohan2019,Miklin2020, Mukherjee2021, Abhyoudai2023}.

As noted previously, self-testing in prepare-measure scenarios typically relies on prior knowledge of the system's dimension, making such approaches inherently semi-DI. It is even argued that this dimensional assumption is essential to enable self-testing \cite{Ambainis2008}. Although a recent scheme has been proposed \cite{Miklin2021} to self-test states and measurements in arbitrary dimensions, it nevertheless requires an upper bound on the dimension of the Hilbert space. As recently highlighted in \cite{Pauwels2022}, real-world devices can inadvertently access higher-dimensional subspaces, potentially undermining conclusions based on dimension-bounded assumptions. This underscores the need for dimension-independent self-testing schemes within the prepare-measure framework.

This work addresses this need by presenting a fully DI self-testing protocol within a constrained prepare-measure scenario. Specifically, we consider the $n$-bit POM task, first introduced by Spekkens et al. \cite{Spekkens2009}, in which parity-oblivious constraint are imposed on Alice's preparations. It has been argued that the parity-obliviousness imposed at the operational level is equivalently represented at the ontic state (hidden variable) level if the ontological model is preparation noncontextual \cite{Spekkens2009}, which is a notion of classicality.  Therefore, demonstrating quantum advantage in the POM  task requires a comparison with the preparation noncontextual ontological model.

We show that the optimal quantum success probability of the POM task exceeds the preparation noncontextual bound. Crucially, our derivation of the optimal quantum bound is independent of the dimension of the underlying system. The requirement of satisfying parity-obliviousness of the inputs plays a crucial role in enabling such a dimension-independent derivation. Thus, we establish a fully DI self-testing protocol in the prepare-measure scenario, thereby advancing the foundational understanding of contextuality and enabling certification of both quantum preparations and measurements devices without any dimensional assumptions.

Moreover, by explicitly constructing the requisite unitary operations, we demonstrate that observing the optimal correlations allows us to map the unknown finite dimensional states and measurements onto an equivalent realisation within a fully characterised finite-dimensional quantum system using the aforementioned constructed unitary. This mapping provides a correspondence between the theoretical DI bound and physically realizable quantum strategies, thereby laying a foundation for a self-testing framework that is both operationally meaningful and implementable  within constrained prepare-measure scenarios.

The paper is organised as follows. In Sec.~\ref{pncpom}, we introduce the conceptual framework of preparation noncontextual ontological models and describe the $n-$bit POM task. Sec.~\ref{PORAC} formalises the constrained prepare-measure communication scenario and establishes the role of parity oblivious constraint in deriving the preparation noncontextual bound of the success probability. In Sec.~\ref{optqn}, we derive the optimal quantum success probability for the $n-$bit POM task using a dimension-independent approach. Sec.~\ref{stspm} presents the self-testing statements and provides the explicit construction of an unitary that maps the DI description to a known finite-dimensional quantum realisation. Finally, in Sec.~\ref{conclu}, we summarise our findings and discuss future directions.

\section{Preparation noncontextuality and POM task} \label{pncpom}

We briefly encapsulate the notion of preparation non-contextuality in an ontological model of quantum theory and the POM task in a prepare-measure scenario. We invoke a modern  framework of an ontological model \cite{spekkens2005,harrigan2010} of quantum theory to introduce the notion of preparation noncontextuality.  In quantum theory, a preparation procedure $(P)$ produces a density matrix $\rho$ and measurement procedure $(M)$ is in general characterized by positive-operator-valued-measure $(E_k$). The probability of a particular outcome $ k $ is determined by the Born rule, that is $p(k|P, M)=\Tr[\rho E_{k}]$. 

In an ontological model of quantum theory, the preparation procedure $P$ prepares a probability distribution $\mu_{P}(\lambda|\rho)$ in ontic state space $\Lambda$ with $\int _\Lambda \mu_{P}(\lambda|\rho)d\lambda=1$ where $\lambda \in \Lambda$. Given a $\lambda$, the probability of obtaining an outcome $k$ is associated with a response function $\xi_{M}(k|\lambda, E_{k}) $ that satisfies $\sum_{k}\xi_{M}(k|\lambda, E_{k})=1$ where a measurement operator $E_{k}$ is realized through $M$. A valid ontological model must reproduce quantum theory, \textit{i.e.}, $\forall \rho $, $\forall E_{k}$ and $\forall k$, $\int _\Lambda \mu_{P}(\lambda|\rho) \xi_{M}(k|\lambda, E_{k}) d\lambda =\Tr[\rho E_{k}]$ must be satisfied.
 
 The notion of preparation noncontextuality in an ontological model emerges from the following statement \cite{spekkens2005}. Two equivalent experimental procedures in quantum theory are assumed to be equivalently represented in an ontological model.
Then, an ontological model of quantum theory is assumed to be preparation non-contextual if 
\begin{align}
\label{ass}
 \forall M, \  k: \  p(k|P, M)=p(k|P^{\prime},M)	\Rightarrow \mu_{P}(\lambda|\rho)=\mu_{P^{\prime}}(\lambda|\rho)
\end{align}
where $\rho$ is prepared by two distinct preparation procedures $ P $ and $ P^{\prime}$ \cite{spekkens2005,Spekkens2009}. 


\section{A Constrained prepare-measure communication game}\label{PORAC}
We consider the $n-$bit POM task, first introduced by Spekkens \textit{et. al.} \cite{Spekkens2009}. It uses the framework of $n\rightarrow 1$ random access code but there is a fundamental difference. As clearly stated in \cite{Spekkens2009}, in the $n$-bit POM task, the constraint in the form of parity obliviousness on Alice's preparations is imposed rather than constraining the information-carrying capacity of the transmitted system.

In a multiplexing task, Alice receives an input from a uniform random set of $n-$bit strings $x^{\delta}\in\{0, 1\}^{n}$ with $\delta\in \{0,1, 2,\dots,2^{n}-1\}$. For example, $x^0 = 00\cdots 00$, $x^1 =
00\cdots01$, $\dots$, $x^{2^{n}-1}=11\cdots 11$. Upon receiving the input $x^\delta$ Alice uses preparation procedure $P_{x^\delta}$ to prepare the state and sends it to Bob. On the other hand, Bob randomly receives an input $y\in\{1,2,....n\} $ and performs a dichotomic measurement $B_y$ which produces output $b\in\{+1,-1\}$. The winning condition of the game is $b=x_{y}^{\delta}$, \textit{i.e.}, Bob has to correctly guess the $y^{th}$ bit of Alice's string $x^\delta$. However, in a POM task, additional constraint on preparations is imposed, by first defining a parity set \cite{Spekkens2009}
\begin{equation}\label{porac set}
 \mathbb{P}_n = \qty{x^\delta\big|x^\delta\in\{0,1\}^n,\sum_{r}x^\delta_r \geq 2}  
\hspace{0.2cm} \text{with} \hspace{0.2cm} r \in \{1,2,\dots,n\}       
\end{equation}
Now, the parity-oblivious condition is that for any element $s \in \mathbb{P}_n$, the information about $x^\delta.s$  must remain oblivious to Bob, where $x^\delta.s = \oplus_i x^\delta_is_i$. Here, $\oplus$ represents the bitwise XOR operation. 

The parity oblivious condition in an operational theory dictates \cite{Ghorai2018}
\begin{equation} \label{poprobn}
\forall s,y,b \ \sum\limits_{x^{\delta}|x^{\delta}.s=0}p(P_{x^{\delta}}|b,B_y)=\sum\limits_{x^{\delta}|x^{\delta}.s=1}p(P_{x^{\delta}}|b,B_y)
\end{equation}
Using the Bayes rule one can write 
 
\begin{align}
\label{po111}
\forall s,y, b \ \ \sum\limits_{x^{\delta}|x^{\delta}.s=0}p(b|P_{x^{\delta}},B_{y}) = \sum\limits_{x^{\delta}|x^{\delta}.s=1} p(b|P_{x^{\delta}},B_{y}) 
\end{align}
This means that the two preparation procedures $P_{{x}^{\delta}|x^{\delta}.s=0}$ and $P_{{x}^{\delta}|x^{\delta}.s=1}$ cannot be distinguished by any outcome $b$ and any measurement $B_{y}$. Such an operational equivalence falls under the premise of preparation non-contextuality given in Eq. (\ref{ass}). Assuming preparation non-contextuality in an ontological model of the above operational theory we can write  
\begin{align}
\label{pnc}
\forall s, \lambda \ \ \sum\limits_{x^{\delta}|x^{\delta}.s=0}\mu(\lambda|P_{x^{\delta}}) =  \sum\limits_{x^{\delta}|x^{\delta}.s=0} \mu(\lambda|P_{x^{\delta}}) 
\end{align}
where $\lambda \in \Lambda$ is the ontic state and $\Lambda$ is the ontic state space. Using Baye's rule we have 
\begin{align}
\label{pnc}
\forall s, \lambda \ \ \sum\limits_{x^{\delta}|x^{\delta}.s=0}\mu(P_{x^{\delta}}|\lambda) =  \sum\limits_{x^{\delta}|x^{\delta}.s=1} \mu(P_{x^{\delta}}|\lambda) 
\end{align}
which implies that for a preparation noncontextual model, the satisfaction of parity obliviousness condition in an operational theory provides equivalent representation at the level of the ontic state \cite{Spekkens2009, Ghorai2018, Pan2021}. In other words, even knowing the ontic state $\lambda$, the parity information of the inputs remain oblivious. 

The success probability of the POM task can be expressed as
\begin{equation}\label{succprob}
       \mathcal{S}_n = \frac{1}{2^nn} \sum_{y \in \{1,2,3,\dots,n\} \atop x^\delta \in \{0,1\}^n} p(b = x^\delta_y | P_{x^\delta},B_y)   
\end{equation}
In an ontological model, the success probability is
\begin{equation}\label{ontsuc}
       (\mathcal{S}_n)_{C} = \frac{1}{2^n n} \int_{\Lambda}\sum_{y \in \{1,2,3,\dots,n\} \atop x^\delta \in \{0,1\}^n} \mu(\lambda|P_{x^\delta})\xi(b|\lambda,B_y) d\lambda  
\end{equation}
In \cite{Spekkens2009}, the classical bound was derived from information theoretic perspective by proving that the parity-oblivious constraint restrict Alice's communication to be one bit. This gives optimal classical success probability as
\begin{equation}
\label{clprob}
(\mathcal{S}_n)_C \leq \frac{1}{2}\qty(1 + \frac{1}{n})
\end{equation}
However, we continue to consider the preparation noncontextuality as a notion of classicality. We provide an explicit derivation of the success probability in a 
 preparation noncontextual   ontological model for $n=2$ and $n=3$ in Appx.~\ref{copt23}, which matches the bound in Eq. (\ref{clprob}). Following that one can straightforwardly derive Eq. (\ref{clprob}) for any arbitrary $n$.

The quantum strategy involves Alice encoding her $n-$bit string $x^\delta$ into a density matrix $\rho_{x^{\delta}}\in \mathscr{L}(\mathcal{H}^d)$ where $d$ is arbitrary and sending it to Bob. The parity oblivious constraint on the probability in Eq.~(\ref{po111}) translates into a condition on Alice's state preparations
\begin{equation} \label{postaten}
     \forall s \ \sum\limits_{x^{\delta}|x^{\delta}.s=0}\rho_{x^{\delta}}=\sum\limits_{x^{\delta}|x^  {\delta}.s=1}\rho_{x^{\delta}}
\end{equation}
Upon receiving the quantum system, Bob performs a dichotomic measurement $B_y\equiv\{\Pi^b_y\}\in \mathscr{L}(\mathcal{H}^d)$ that produces an outcome $b\in\{+1,-1\}$. The quantum success probability can then be written as \begin{equation} \label{succquant}
        (\mathcal{S}_n)_Q = \frac{1}{2^nn} \sum_{y \in \{1,2,3,\dots,n\} \atop x^\delta \in \{0,1\}^n}  \Tr[\rho_{x^{\delta}} \Pi^b_y]
\end{equation}
We prove that the optimal quantum success probability to be
\begin{equation}
    (\mathcal{S}_n)_Q^{opt} = \frac{1}{2}\qty(1+\frac{1}{\sqrt{n}})
\end{equation}
 without assuming the dimension of the quantum system. This implies that a preparation noncontextual model cannot reproduce optimal quantum success probability. Achieving $(\mathcal{S}_n)_Q^{opt}$ allows us to DI self-test Alice's preparations and Bob's measurements. 

\section{Dimension-independent derivation of optimal quantum success probability} \label{optqn}

We explicitly derive the optimal quantum success probability of the $n-$bit POM task for the cases $n=2$ and $n=3$, without assuming dimension. Further, we extend our proof to any arbitrary $n$.

\subsection{Optimal quantum success probability for $2-$bit POM task} \label{optq2}

In a $2-$bit POM task, Alice randomly receives a string $x^{\delta}\in\{00,01,10,11\}$, upon which she prepares the corresponding quantum state $\rho_{x^\delta}$, and sends them to Bob. Given that the parity set comprises a single element, $\mathbb{P}_2=\{11\}$, we observe that for $x^{\delta}\in \{00,11\}$, the parity condition satisfies $x^\delta.s = 1$, while for $x^{\delta}\in\{01,10\}$, we have $x^\delta.s = 0$. This means that to ensure parity obliviousness, the following condition must hold $\forall y$ 
\begin{equation}\label{po1}
p(b|\rho_{00},B_y) + p(b|\rho_{11},B_y)=p(b|\rho_{01},B_y) + p(b|\rho_{10},B_y) \end{equation}
In quantum theory, this translates into
\begin{equation} \label{postate2}
   s=11: \ \ \  \rho_{x^0}+\rho_{x^3}=\rho_{x^1}+\rho_{x^2}
\end{equation}
Upon receiving the system, Bob randomly performs one of two measurements, $B_1$ or $B_2$. From Eq.~(\ref{succquant}), the quantum success probability $(\mathcal{S}_2)_Q$ is evaluated as
\begin{equation}\label{2bitsuccp}
\begin{aligned}
    (\mathcal{S}_2)_Q &= \frac{1}{8} \sum_{y \in \{1,2\} \atop x^\delta\in\{0,1\}^2 } \Tr[\rho_{x^\delta} \Pi_{B_y}^{x^\delta_y}] \\
&=\frac{1}{2} + \frac{1}{16} \Tr\Bigr[\qty(\rho_{x^0} + \rho_{x^1} - \rho_{x^2} - \rho_{x^3})B_1 \\
    & \ \ \ \ \ \ \ \ \ + \qty(\rho_{x^0} - \rho_{x^1} + \rho_{x^2} - \rho_{x^3})B_2  \Bigr]  
\end{aligned}
\end{equation}
Here we use $\Pi_{B_y}^b = \frac{1}{2}\qty(\openone+(-1)^bB_y)$, by exploiting the fact that $B_y$ is dichotomic with eigenvalues $\pm 1$. Moreover, since there is no restriction on the dimension, an extension of Naimark's dilation theorem \cite{Irfan2020} allows us to consider Bob’s measurement to be a projective measurement and can always represent the measurement as $B_y=2\Pi_{B_y}^b-\openone_d$ with $\Pi_{B_y}^b$ acting in some finite-dimensional Hilbert space $\mathcal{H}$.

At first glance of Eq. (\ref{2bitsuccp}), it may appear that the maximal success probability $(\mathcal{S}_2)_Q=1$ could be attainable in some arbitrary dimension $d$, if $\rho_{x^0}$,  $\rho_{x^1}$,  $\rho_{x^2}$ and $\rho_{x^3}$ are mutual eigenstates of $B_1$ and $B_2$, associated with appropriate eigenvalues. However, such a construction would not fulfil the parity-oblivious constraint in Eq.~(\ref{postate2}).

We consider that Alice’s state preparations in an arbitrary dimension $d$ which comply with the parity-obliviousness condition in Eq.~(\ref{postate2}), take the following form
\begin{equation}\label{rhoai}
\begin{aligned}
    &\rho_{x^0} = \frac{\openone_d + A_0}{d}, \hspace{0.5cm} \rho_{x^3} = \frac{\openone_d + A_3}{d} \hspace{0.5cm} \\
    &\rho_{x^1} = \frac{\openone_d + A_1}{d}, \hspace{0.5cm} \rho_{x^2} = \frac{\openone_d + A_2}{d}
\end{aligned}
\end{equation}
where $A_\delta \in \mathscr{L}(\mathcal{H}^d)$ are arbitrary dichotomic and traceless Hermitian operators defined over the linear space of operators on a $d-$dimensional Hilbert space, satisfying $\Tr[A_\delta]=0$ and $A^2_\delta=\openone_d$. For example, when $d=4$, each $A_\delta$ can be expressed as a linear combination of up to fifteen observables. Substituting relations in Eq.~(\ref{rhoai}) into Eq.~(\ref{2bitsuccp}) and simplifying, we find
\begin{equation} \label{succ1}
    \begin{aligned}
        (\mathcal{S}_2)_Q
            &= \frac{1}{2}+\frac{1}{16d}\Tr[(A_0-A_3)\qty(B_1+B_2)+(A_1-A_2)\qty(B_1-B_2)]
    \end{aligned}  
\end{equation}
Now, we define
\begin{equation} \label{normobs}
\begin{aligned}
        &\mathcal{B}_{0} = \frac{B_1 +B_2}{\omega_0}, \ \ \mathcal{B}_{1} = \frac{B_1 -B_2}{\omega_1} \\
        &\mathcal{A}_{0} = \frac{A_0 -A_3}{\alpha_0}, \ \ \mathcal{A}_{1} = \frac{A_1 -A_2}{\alpha_1}
\end{aligned}
\end{equation}
where $\omega_i$ are operator norms defined as follows
 \begin{equation} \label{onorm}
 \begin{aligned}
     &\omega_i = \norm{B_1+(-1)^i B_2} \ \ \forall i \in \{0,1\} \\
     &\alpha_0 = \norm{A_0-A_3} \ ;\ \alpha_1 = \norm{A_1-A_2}
 \end{aligned}
        \end{equation}
The operator norm is defined as $\norm{O}= \sup\limits_{v \in H, \norm{v}=1} \norm{Ov}$. For Hermitian operators, this corresponds to the spectral radius, i.e. the largest absolute value of its eigenvalues, $\norm{O}=\max \{|\lambda_{min}|,|\lambda_{max}|\}$. Substituting Eq.(\ref{normobs}) into Eq.(\ref{succ1}), the quantum success probability simplifies to
\begin{eqnarray}
(\mathcal{S}_2)_Q &=& \frac{1}{2} + \frac{1}{16d} \Tr[\alpha_0\omega_0 \mathcal{A}_0\mathcal{B}_{0}  + \alpha_1\omega_1 \mathcal{A}_1\mathcal{B}_{1}]
\end{eqnarray}
since $\omega_i$ and $\alpha_i$ are positive real numbers, the success probability is maximised when $\mathcal{A}_i\mathcal{B}_i = \openone_{d}$, which implies that $\mathcal{A}_0 = \mathcal{B}_0$ and $\mathcal{A}_1 = \mathcal{B}_1$.

Given the parity-obliviousness condition $\rho_{x^0}+\rho_{x^3} = \rho_{x^1}+\rho_{x^2}$ which implies $A_0 + A_3 = A_1 + A_2$, squaring both sides of the resultant relation yields $\{A_0,A_3\}= \{A_1,A_2\}$. Hence, the values of $\alpha_k$ is determined as 
\begin{equation}
    \begin{aligned}
        \alpha_0^2 &= \norm{A_0-A_3}^2  = \norm{2\openone-\{A_0,A_3\}}=\alpha_1^2
    \end{aligned}
\end{equation}
Since $\alpha_k>0$, it follows that $\alpha_0=\alpha_1=\alpha$, which results in the optimal success probability
\begin{eqnarray}\label{maxai}
    (\mathcal{S}_2)_Q^{opt}&=& \frac{1}{2} + \frac{1}{16}\qty[\max\limits_{\omega_i,\alpha} \ (\omega_0\alpha +\omega_1\alpha)] \\
    &\leq& \frac{1}{2} + \frac{1}{8\sqrt{2}}\qty[ \max\limits_{\alpha,\omega_i} \alpha\sqrt{\qty(\omega^2_0+\omega^2_1)}]
\end{eqnarray}
This inequality follows from Jensen’s inequality, as well as the Cauchy-Schwarz inequality $\Vec{a}.\Vec{b}\leq |\Vec{a}||\Vec{b}|$, applied to $\Vec{a}=(\omega_0,\omega_1)$ and $\Vec{b}=(1,1)$ \cite{Boyd2004}. Evaluating
\begin{equation}
    \begin{aligned}
        \omega_0^2+\omega_1^2 &= \norm{B_1+B_2}^2+\norm{B_1-B_2}^2\\
        &= \norm{2\openone_d+\{B_1,B_2\}}+\norm{2\openone_d-\{B_1,B_2\}}=4
    \end{aligned}
\end{equation}
This equality holds for any unitary Hermitian operator $B_1$ and $B_2$. The maximum value of $\alpha=\sqrt{\norm{2\openone_d -\{A_0,A_3\}}}$ is attained when $\{A_0,A_3\}=-2\openone_d$, implying $A_0+A_3=A_1+A_2=0$ and $\alpha=2$. Therefore, the optimal quantum success probability becomes
\begin{equation}
    (\mathcal{S}_2)_Q^{opt}=\frac{1}{2}\qty(1+\frac{1}{\sqrt{2}})
\end{equation}
Noting that $\omega_0+\omega_1\leq\sqrt{2(\omega_0^2+\omega_1^2)}$, equality is obtained when $\omega_0=\omega_1$, from which we deduce that $\{B_1,B_2\}=0$, and hence $\omega_0=\omega_1=\sqrt{2}$. The optimal conditions are then 
\begin{equation}
    \mathcal{A}_i=\mathcal{B}_i, \ A_0+A_3=A_1+A_2=0, \ \{B_1,B_2\}=0
\end{equation}
which then leads to
\begin{equation}
    \begin{aligned}
        \frac{A_0-A_3}{\sqrt{2}} &= B_1+B_2\ ; \ \frac{A_1-A_2}{\sqrt{2}}  = B_1-B_2
    \end{aligned}
\end{equation}
Accordingly, Alice's optimal preparations are 
\begin{equation}
    \begin{aligned}
        \rho_{00} &= \frac{1}{d}\qty(\openone_d+\frac{B_1+B_2}{\sqrt{2}})\ ; \ \  \rho_{11} = \frac{1}{d}\qty(\openone_d-\frac{B_1+B_2}{\sqrt{2}}) \\
        \rho_{01} &= \frac{1}{d}\qty(\openone_d+\frac{B_1-B_2}{\sqrt{2}}) \ ; \ \
        \rho_{10} = \frac{1}{d}\qty(\openone_d-\frac{B_1-B_2}{\sqrt{2}})  
    \end{aligned}
\end{equation}
satisfying the parity-oblivious condition $\rho_{00}+\rho_{11} = \rho_{01}+\rho_{10} = 2\openone/d$.

\subsection{Optimal success probability in the $3-$bit POM task}\label{porac3}

In the $3-$bit POM task, Alice prepares one of eight quantum states denoted by $\rho_{x^\delta}$, with $x^\delta\in\{0,1\}^3$ and sends it to Bob. The parity-obliviousness constraint for this task is provided in Eq.~(\ref{poprobn}). In the quantum theory, the constraint in Eq.~(\ref{poprobn}) can be expressed as follows
\begin{eqnarray}\label{postate3}
    s=011&:& \rho_{000} + \rho_{011} + \rho_{100} + \rho_{111} = \rho_{110} + \rho_{101} + \rho_{010} + \rho_{001} \nonumber \\
    s=101&:& \rho_{000} + \rho_{001} + \rho_{110} + \rho_{111} = \rho_{010} + \rho_{100} + \rho_{011} + \rho_{101} \nonumber  \\
    s=110&:& \rho_{000} + \rho_{010} + \rho_{101} + \rho_{111} = \rho_{001} + \rho_{100} + \rho_{011} + \rho_{110} \nonumber  \\
    s=111&:& \rho_{000} + \rho_{011} + \rho_{101} + \rho_{110} = \rho_{001} + \rho_{010} + \rho_{100} + \rho_{111} \nonumber \\
\end{eqnarray}
Upon receiving the state, Bob performs one of three measurements $B_1$, $B_2$ or $B_3$. Using Eq.~(\ref{succquant}) and taking Bob's measurement to be projective, i.e., $\Pi^{x^\delta_y}_{B_y} = \frac{1}{2}\qty(\openone+(-1)^{x^\delta_y}B_y)$, the quantum success probability $(\mathcal{S}_3)_Q$ is given by    
\begin{equation}
 (\mathcal{S}_3)_Q = \frac{1}{2} +  \frac{1}{48}\sum_{y \in \{1,2,3\} \atop x^\delta \in \{0,1\}^3} (-1)^{x^\delta_y}\Tr[\rho_{x^\delta} B_y]
\end{equation}
Following a procedure analogous to that for the $2-$bit POM task, we define Alice's preparations in the form
\begin{equation}\label{neq3rho}
\begin{aligned}
    \rho_{x^{\delta}} = \frac{\openone + A_{\delta}}{d}, \ \ \ \ \delta\in\{0,1,2,\dots,7\}
\end{aligned}
\end{equation}
where $A_\delta \in \mathscr{L}(\mathcal{H}^d)$ are arbitrary dichotomic and traceless Hermitian operators. The quantum success probability can then be rewritten as
\begin{equation}
    \begin{aligned}
        (\mathcal{S}_3)_Q =  \frac{1}{2} &+ \frac{1}{48d} \sum_{k=0}^{3}\Tr[\omega_k\alpha_k\mathcal{A}_k\mathcal{B}_k ]
    \end{aligned}
\end{equation}
Here, the normalised operators $\mathcal{A}_k$ and $\mathcal{B}_k$ are defined as
\begin{equation}
\begin{aligned}
&\mathcal{A}_0 = \frac{A_0-A_7}{\alpha_0}, \ \ \mathcal{B}_{0} = \frac{B_1 + B_2 + B_3}
{\omega_0}\\
&\mathcal{A}_1=\frac{A_1-A_6}{\alpha_1}, \ \ \mathcal{B}_{1} = \frac{B_1 + B_2 -B_3}{\omega_1}, \\ 
&\mathcal{A}_2 = \frac{A_2-A_5}{\alpha_2}, \ \ \mathcal{B}_{2} = \frac{B_1 - B_2 + B_3}{\omega_2} \\
&\mathcal{A}_3 = \frac{A_3-A_4}{\alpha_3}, \ \ \mathcal{B}_{3} = \frac{B_1 - B_2 - B_3}{\omega_3}
\end{aligned}
\end{equation}
with corresponding normalisation factors $\omega_k, \alpha_k>0$, defined in terms of operator norm as
 \begin{equation}
     \begin{aligned}
         \omega_0 &=  \norm{B_1+B_2+B_3}\ ; \ \ \ 
         \omega_1 =  \norm{B_1+B_2-B_3}\\
         \omega_2 &=  \norm{B_1-B_2+B_3}\ ; \ \ \
         \omega_3 = \norm{B_1-B_2-B_3}\\
         \alpha_0 &=  \norm{A_0-A_7}\ ; \ \ \         \alpha_1 =  \norm{A_1-A_6}\\
         \alpha_2 &=  \norm{A_2-A_5}\ ; \ \ \ 
         \alpha_3 =  \norm{A_3-A_4}
     \end{aligned}
 \end{equation}
Using the parity-obliviousness relations in Eq.~(\ref{postate3}), we obtain
\begin{equation}\label{poobs3}
    \begin{aligned}
        &s=011: \ A_0 + A_3 + A_4 + A_7 = A_6 + A_5 + A_2 + A_1 \\
    &s=101: \ A_0 + A_1 + A_6 + A_7 = A_2 + A_4 + A
    _3+ A_5
    \\
    &s=110: \ A_0 + A_2 + A_5 + A_7 = A_1 + A_4 + A_3 + A_6
    \\
    &s=111: \ A_0 + A_3 + A
    _5+ A_6 = A_1 + A_2 + A_4 + A_7  
    \end{aligned}
\end{equation}
From this set of constraints given by Eq.~(\ref{poobs3}), it follows that (see Appx.~\ref{apnA})
\begin{equation} \label{acpo3}
    \{A_0,A_7\}=\{A_1,A_6\}=\{A_2,A_5\}=\{A_3,A_4\}
\end{equation}
Hence, $\alpha_i^2 = \alpha_j^2 \ \forall i \neq j\in\{0,1,2,3\}$.  Since $\alpha_i>0$, we have $\alpha_i = \alpha_j = \alpha$. To maximise the quantum success probability, we require $\mathcal{A}_k\mathcal{B}_{k} =\openone_d \ \forall k$. Applying the Cauchy–Schwarz inequality $\Vec{a}.\Vec{b}\leq |\Vec{a}||\Vec{b}|$ to $\Vec{a}=(\omega_0,\omega_1,\omega_2,\omega_3)$ and $\Vec{b}=(1,1,1,1)$, we find that the optimal quantum success probability satisfies
\begin{equation}
\begin{aligned}
   \qty(\mathcal{S}_3)^{opt}_Q \leq  \frac{1}{2} + \frac{1}{24} \qty[\max_{\alpha, \omega_i}\alpha\sqrt{\omega_0^2 + \omega_1^2 + \omega_2^2 + \omega_3^2}]
\end{aligned}
\end{equation}
Now, evaluating $\sum_{i=0}^3\omega_i^2$, we obtain
\begin{equation}
    \begin{aligned}
        \omega_0^2&=\norm{3 \openone + \{B_1,B_2\}+\{B_2,B_3\}+\{B_2,B_3\}} \\
        \omega_1^2&=\norm{3 \openone + \{B_1,B_2\}-\{B_2,B_3\}-\{B_1,B_3\}}\\
        \omega_2^2&=\norm{3 \openone - \{B_1,B_2\}-\{B_2,B_3\}+\{B_1,B_3\}} \\
        \omega_3^2&=\norm{3 \openone - \{B_1,B_2\}+\{B_2,B_3\}-\{B_1,B_3\}}
    \end{aligned}
\end{equation}
Equality in the Cauchy–Schwarz bound, holds if, $\omega_0=\omega_1=\omega_2=\omega_3$, which implies that $\{B_y,B_{y'}\} = 0 \ \forall y \neq y' \in[3]$. In this scenario, $\omega_i = \sqrt{3}$, and thus $\sum_{i=0}^3\omega_i^2 = 12$. The maximum value of $\alpha = \sqrt{\norm{2\openone_d - \{A_0,A_7\}}}$ is achieved when $\{A_0,A_7\}=-2\openone_d$,  which satisfies the parity-obliviousness constraints of Eq.~(\ref{poobs3}), with both sides evaluating to zero. Consequently, $\alpha = 2$, and the optimal quantum success probability becomes
\begin{equation}\label{optsucc3to1}
(\mathcal{S}_3)_Q^{opt} = \frac{1}{2}\Biggl(1 + \frac{1}{\sqrt{3}}\Biggl)
\end{equation}
The optimality condition $\{A_0,A_7\} = -2\openone\implies A_0=-A_7$, and from Eq.~(\ref{acpo3}), we set $A_1=-A_6$, $A_2=-A_5$, $A_3=-A_4$. This finally leads to the following preparations of Alice
\begin{equation}
    \begin{aligned}
        \rho_{000} &= \frac{1}{d}\qty(\openone_d+\frac{B_1+B_2+B_3}{\sqrt{3}}), \\
        \rho_{001} &= \frac{1}{d}\qty(\openone_d+\frac{B_1+B_2-B_3}{\sqrt{3}})\\
        \rho_{010} &= \frac{1}{d}\qty(\openone_d+\frac{B_1-B_2+B_3}{\sqrt{3}})\\
        \rho_{011} &= \frac{1}{d}\qty(\openone_d+\frac{B_1-B_2-B_3}{\sqrt{3}})
    \end{aligned}
\end{equation}
The remaining states (e.g., $\rho_{111}$) can be obtained by exploiting their orthogonality with the corresponding complements (e.g.,$\rho_{000}$), ensuring the overall state ensemble conforms to the parity-oblivious condition in Eq. (\ref{postate3}).

\subsection{Optimal success probability in the  $n-$bit POM task} \label{optqnc}

For $n-$bit POM task, Alice prepares $2^n$ distinct states, denoted by $\rho_{x^{\delta}}$, corresponding to $n-$bit string $x^{\delta}\in \{0,1\}^n$.  The parity set is defined in Eq.~(\ref{porac set}) and the parity oblivious conditions can then be derived as follows. 
Considering each state is of the form $\rho_{x^{\delta}}=\frac{1}{d}\qty(\openone+A_\delta)$, with $A_\delta$ being traceless, dichotomic Hermitian operators, the condition follows from Eq.~(\ref{postaten}) becomes \cite{Ghorai2018}
\begin{equation}\label{ncond}
    \forall s \in \mathbb{P}_n \  \  \sum_{x^{\delta}}(-1)^{x^{\delta}\cdot s} A_{\delta} =0
\end{equation}
 The quantum success probability is given by
\begin{equation}
\begin{aligned}
    (\mathcal{S}_n)_Q =\frac{1}{2} + \frac{1}{2^{n+1}n}\sum_{y \in [n] \atop x^\delta\in\{0,1\}^n } (-1)^{x^\delta_y}\Tr[\rho_{x^\delta}B_y]
\end{aligned}
\end{equation}
which needs to be calculated by satisfying Eq.~(\ref{ncond}). Here we again use $\Pi_{B_y}^{x^\delta_y}=\frac{1}{2}(\openone+(-1)^{x^\delta_y}B_y)$. Considering each quantum state takes the form $\rho_{x^\delta} = \frac{1}{d}\qty(\openone_d + A_\delta)$, where $A_\delta\in\mathscr{L}(\mathcal{H}^d)$ are traceless, dichotomic hermitian operators, the expression for the quantum success probability simplifies to
\begin{equation}\label{n-scenario}
\begin{aligned}
     (\mathcal{S}_n)_Q &= \frac{1}{2} + \frac{1}{2^{n+1}nd} \Tr\qty[\sum_{\delta=0}^{2^{n-1}-1}\alpha_\delta\omega_\delta\mathcal{A}_\delta\mathcal{B}_\delta] 
\end{aligned}
\end{equation}
The operators $\mathcal{A}_\delta$, $\mathcal{B}_\delta$ and the corresponding normalisation factors are defined as
\begin{equation}
\begin{aligned}
&\mathcal{B}_\delta = \frac{\sum_{y = 1}^n (-1)^{x_y} B_y}{\omega_\delta} \ ; \ \ \ \mathcal{A}_\delta = \frac{A_\delta+A_{\overline{\delta}}}{\alpha_\delta} \\
&\omega_\delta = \norm{\qty(\sum_{y = 1}^n (-1)^{x^\delta_y} B_y)}\ ; \ \alpha_\delta = \norm{\qty(A_\delta-A_{\overline{\delta}})}
\end{aligned}
\end{equation}
Here $\overline{\delta}$ is the decimal representation of the bitwise complement of $x^\delta \equiv x_1x_2\dots x_n$. By considering the full set of parity-oblivious constraints, we show that (see Appx.~\ref{apnB})
\begin{equation} \label{acpon}
    \{A_\delta,A_{\overline{\delta}}\}=\{A_\eta ,A_{\overline{\eta}}\} \ ; \forall \delta \neq \eta
\end{equation}
This condition ensures $\alpha_\delta = \alpha_\eta = \alpha \ \forall \delta \neq \eta$. To maximise the quantum success probability, it suffices that $\mathcal{A}_k\mathcal{B}_{k} =\openone_d \ \forall k$. 
This leads to the optimal bound
\begin{equation}
\begin{aligned}
     (\mathcal{S}_n)_Q^{opt} &= \frac{1}{2} + \frac{1}{2^{n+1}n} \qty(\max_{\alpha, \omega_i}\sum_{\delta=0}^{2^{n-1}-1}\alpha\omega_\delta) \\ 
     &\leq\frac{1}{2} + \frac{1}{2^{n+1}n} \max_{\alpha, \omega_i}\qty(\alpha\sqrt{2^{n-1}\sum_{i=0}^{2^{n-1}-1} \omega_i^2 })
\end{aligned}
\end{equation}
The inequality in the second line follows from the Cauchy–Schwarz inequality applied to the vectors $\Vec{a}=(\omega_0,\omega_1,\dots, \omega_{2^{n-1}-1})$ and $\Vec{b}=(1,1,\dots,1)$. Equality is attained \textit{iff} $\omega_i=\omega_j \ \forall i \neq j$, which is the case precisely when all the observables $\{B_y,B_{y'}\} = 0 \ \forall y \neq y'$ (see Appx.~\ref{apnC}). Under these conditions, each $\omega_i = \sqrt{n}$, giving $\sum_{\delta}\omega_\delta^2 = n.2^{n-1}$. Additionally, the maximal value of $\alpha = \sqrt{\norm{2\openone_d - \{A_\delta,A_{\overline{\delta}}\}}}=2$, attained when $\{A_\delta,A_{\overline{\delta}}\}=-2\openone_d$, i.e., $A_\delta = -A_{\overline{\delta}}$. 

Under these optimal conditions, the parity-oblivious constraints in Eq.~(\ref{postaten}) are satisfied, and the optimal quantum success probability becomes
\begin{equation}
    \begin{aligned}
        (\mathcal{S}_n)_Q^{opt} = \frac{1}{2}\qty(1+\frac{1}{\sqrt{n}})
    \end{aligned}
\end{equation}
This clearly outperforms the classical optimal success probability $(\mathcal{S}_n)^{opt}_C = \frac{1}{2} \left(1 + \frac{1}{n} \right)$

These optimality conditions further lead to the following self-testing implications - (i) The quantum states encoding binary strings that are bitwise complements of each other are pairwise orthogonal. (ii) The observables $B_1$, $B_{2}$, \dots, $B_{n}$ are mutually anticommuting. Additionally, one can construct Alice's preparation as follows 
\begin{equation}\label{astate}
    \rho_{x^\delta} = \frac{1}{d}\qty(\openone_d+\frac{\sum_{y = 1}^n (-1)^{x^\delta_y} B_y}{\sqrt{n}} )
\end{equation}
Thus, we conclude that in any generalised $n-$bit POM task, where Alice prepares one of $2^n$ quantum states and Bob performs one of $n$ dichotomic measurements, the maximum achievable success probability exceeds the preparation non-contextual bound. This maximal quantum violation provides a DI self-testing for both the prepared quantum states and the measurement observables involved.


\section{DI Self-testing in prepare-measure scenario}\label{stspm}

In the preceding Sec.~\ref{optqn}, we establish that to obtain the optimal quantum advantage of the $n-$bit POM task, following conditions must be satisfied.
\begin{enumerate}[(i)]
    \item Bob's $n$ observables, $B_y \in \mathscr{L}(\mathcal{H}^d)$, must mutually anticommute, \textit{i.e.}, for $2\leq y,y'\leq n$, we require $\{B_y,B_{y'}\}=0$. This condition constrains the minimal Hilbert space dimension. While no general result is known for the maximum number of mutually anticommuting observables in an arbitrary $d$-dimensional space, it is known that for $d=2^m$, there exist at most $2m+1$ mutually anticommuting traceless Hermitian operators with eigenvalues $\pm 1$ \cite{Bandyopadhyay2002}, a result that follows from the structure of Clifford algebras and their minimal faithful representations \cite{Wehner2010, Slofstra2011}. Consequently, Bob's measurement operators must act on a Hilbert space of minimal dimension $d^{*}=2^m$, where $m=\lceil \frac{1}{2}(n-1) \rceil$ \label{i}.
    
    \item Alice must prepare quantum states of at least the same dimension, $d^{*}=2^m$, with $m=\lceil \frac{1}{2}(n-1) \rceil$. If the prepared states lie in a smaller-dimensional Hilbert space, then regardless of Bob’s measurements, the success probability cannot attain its optimal value \label{ii}.

    \item \label{iii}Alice's prepared quantum states $\rho_{x^{\delta}}\in\mathscr{L}(\mathcal{H}^d)$ given by Eq.~(\ref{astate}) result in a family of states exhibiting directional bias that lie entirely within the Clifford subspace spanned by $\{\openone, B_1,\dots,B_n\}$. Each $\rho_{x^{\delta}}$ corresponds to one of the $2^n$ vertices of an $n$-dimensional regular hypercube of edge length $\frac{2}{\sqrt{n}}$, embedded in the Clifford subspace of the generalised Bloch sphere associated with $\mathcal{H}^d$. These vertices lie on the surface of a unit-radius $n$-dimensional sphere, and their coordinates are given by the Clifford Bloch vectors (see Appx.~\ref{Appxcbghs})
    \begin{equation}\label{genbv}
        \Vec{r}_{x^{\delta}}=\frac{1}{\sqrt{n}}\qty((-1)^{x_1^{\delta}},\dots,(-1)^{x_n^{\delta}}) \ \forall x^{\delta}\in(0,1)^n
    \end{equation}
\end{enumerate}
It is important to note that conditions (\ref{i})-(\ref{iii}) are purely mathematical statements concerning optimality. To elevate these to formal self-testing results, we must establish the existence of an unitary that maps the physical experiment to a reference experiment \cite{saha2020, Sarkar2021}. 

We begin by assuming that Bob's measurements $B_y \in\mathscr{L}(\mathcal{H}^d)$ and Alice's preparations $\rho_{x^\delta}$ are elements of the same Hilbert space, and together they generate the correlations $\qty{p(b|y,\rho_{x^{\delta}})=\Tr[\rho_{x^{\delta}}\Pi^b_y]}$, yielding the optimal quantum success probability for the $n$-bit POM task. This constitutes the \textit{physical experiment}. 

Now, consider a \textit{reference experiment}, in which Bob's measurements $\{{B}^\prime_y \in \mathscr{L}(\mathcal{H}^{d^\prime})\}$ and Alice's prepared states $\{\rho^{\prime}_{x^\delta}\in   \mathscr{L}(\mathcal{H}^{d^\prime})\}$ attain the same optimal quantum success probability, with $\mathcal{H}^{d^{\prime}}$ being a Hilbert space of known finite dimension. In our construction, we take $\mathcal{H}^{d^{\prime}}\equiv (\mathcal{H}^2)^{\otimes m}$ for a suitably chosen $m$, depending on the values of $n$ as $m=\lceil \frac{1}{2}(n-1) \rceil$. Attainment of the optimal quantum success probability implies that the physical and reference experiments are equivalent, provided that there exists a unitary establishing the following correspondence
\begin{equation}\label{sts}
\begin{aligned}
     \exists U : \mathcal{H}^d \to \mathcal{H}^{d^\prime}\otimes \mathcal{H}^{J} \ \text{s.t.} \ & (i) \ U B_yU^\dagger = B^{\prime}_y\otimes\openone_{J} \\
     & (ii) \  U\rho_{x^\delta}U^\dagger = \rho^\prime_{x^\delta}\otimes\frac{\openone_{J}}{J} 
\end{aligned}
\end{equation}
where $ B_{y} \in  \mathscr{L}(\mathcal{H}^d)$, ${B}^{\prime}_y=\boldsymbol{{\sigma}}_i \in  \mathscr{L}(\mathcal{H}^{d^{\prime}})$ with $\boldsymbol{\sigma}_i$ being Pauli operators acting on multiple qubits, and $\openone_{J}\in\mathscr{L}(\mathcal{H}^{J})$, satisfying $d = d^{\prime}J$. The reference state $\rho^\prime_{x^\delta}$ is given by
\begin{equation}
    \rho^\prime_{x^\delta} = \frac{1}{d^\prime}\qty(\openone_{d^\prime}+\frac{\sum_{y = 1}^n (-1)^{x^\delta_y} B^\prime_y}{\sqrt{n}} )
\end{equation}
Whenever such a unitary $U$ can be constructed (or its existence established), it implies that the reference experiment is self-tested by the observed correlations of the physical experiment. This leads to the following formal statement.
\begin{thm}
Let a quantum strategy $\qty{\rho_{x^\delta}, {B_y\in  \mathscr{L}(\mathcal{H}^d)}}$, achieve maximum quantum success probability in the $n-$bit POM task, where $\mathcal{H}^d$ is an unknown finite-dimensional Hilbert space. Then, this strategy self-tests the reference preparations and measurements $\qty{\rho^\prime_{x^\delta}, B^\prime_y\in  \mathscr{L}(\mathcal{H}^{d^\prime})}$, upto unitary freedom and complex conjugation, where $\mathcal{H}^{d^\prime}$ is of known dimension, if there exists a unitary operation $U: \mathcal{H}^d \to \mathcal{H}^{d^\prime}$ such that Eq.~(\ref{sts}) satisfied.
\end{thm}

\begin{proof}
We prove this by explicitly constructing such unitaries for the cases of the $3-$bit and $5-$bit POM tasks in the following sections. From these base cases, the result for general $n-$bit POM tasks follows recursively.
\end{proof}

\subsection{Construction of unitary for $3-$bit POM task} \label{st3bit}

We begin directly with the $3-$bit POM task, since the argument for the $2-$bit case will follow analogously. Note that Bob's three observables are dichotomic, traceless operators satisfying $B_y^2 = \openone_{d}$. Therefore, there always exists a basis in which $B_1$ is diagonal, with all positive and negative eigenvalues grouped together. Consequently, there exists a unitary operator $U_1\in \mathscr{L}(\mathcal{H}^d)$ such that  
\begin{equation}
    \begin{aligned}
       \mathtt{B}_1= U_1B_1U_1^{\dag} = \sigma_z\otimes\openone_{J}
    \end{aligned}
\end{equation}
where $\sigma_z$ and $\openone_{J}$ are acting on $\mathcal{H}^2$ and $\mathcal{H}^{J}$, respectively, with $\mathcal{H}^d=\mathcal{H}^2 \otimes \mathcal{H}^{J}$. Without loss of generality, we express the rotated observable $\mathtt{B}_2$ in block form as
\begin{equation}
    \mathtt{B}_2 = U_1B_2U_1^\dag= \begin{bmatrix}
            \eta_{00} & \eta_{01} \\
            \eta_{10} & \eta_{11}
        \end{bmatrix}
\end{equation}
where $\eta_{ij} \in \mathscr{L}(\mathcal{H}^J)$  with $J = d/2$. Using $\{B_1,B_2\}=0 \implies\{\mathtt{B}_1,\mathtt{B}_2\} = 0$, we obtain $\mathtt{B}_2 = -\mathtt{B}_1\mathtt{B}_2\mathtt{B}_1$. Solving this yields $\eta_{00}=\eta_{11}=0$. Furthermore, the Hermiticity of $\mathtt{B}_2$ implies  $\eta_{10}=\eta_{01}^\dagger$. Now as we can see that this leads to an off diagonal form of $B_2$. For notational convinience we introduce a variable $X_2 = \eta_{10}$. Similarly, using the Hermiticity of $\mathtt{B}_3$ and its anticommutation with $\mathtt{B}_1$, we have
\begin{equation}
    \mathtt{B}_2 = \begin{bmatrix}
            0 & X_2 \\
            X_2^\dagger & 0
        \end{bmatrix}; \  \ \ \ \ \ 
        \mathtt{B}_3 = \begin{bmatrix}
            0 & X_3 \\
            X_3^\dagger & 0
        \end{bmatrix}
\end{equation}
where, $X_3 \in \mathscr{L}(\mathcal{H}^J)$. Now consider the anticommutation relation $\{\mathtt{B}_2,\mathtt{B}_3\} =0$. This leads to the following identities
\begin{equation}
    X_3 = - X_2X_3^\dagger X_2; \ \ X_3^\dagger = - X_2^\dagger X_3X_2^\dagger
\end{equation}
Also, since $B_3^2=B_2^2=\openone_d$, it follows that $X_3^\dagger X_2X_3^\dagger X_2 = \iota^2\openone_J$. Defining the unitary operator
\begin{equation}
    U_2 = \begin{bmatrix}
        \openone_{J} & 0\\
        0 & \iota X_2
    \end{bmatrix}
\end{equation}
Note that $U_2\qty(\sigma_z\otimes\openone_{J})U_2^\dagger = \sigma_z\otimes\openone_{J}$. Applying $U_2$ to $\mathtt{B}_2$, we obtain 
\begin{equation}
    U_2\mathtt{B}_2U_2^\dagger = \sigma_x\otimes\openone_{J}
\end{equation}
and, for $\mathtt{B}_3$, we get
\begin{equation}
    U_2\mathtt{B}_3U_2^\dagger    = \begin{bmatrix}
            0 & X_3X_2^\dagger\\
            X_2X_3^\dagger & 0
        \end{bmatrix}
\end{equation}
From the relations derived earlier, we have $X_2X_3^\dagger = \iota\openone_{J}$ and hence $X_3X_2^\dagger = -\iota \openone_{J}$. Therefore, $U_2\mathtt{B}_3U^\dagger_2=\sigma_y\otimes\openone_{J}$.
Thus, there exists $U:=U_2U_1$ such that
\begin{equation}
    UB_1U^\dag = \sigma_z\otimes\openone_{J}, \ UB_2U^\dag = \sigma_x\otimes\openone_{J}, \ UB_3U^\dag = \sigma_y\otimes\openone_{J} 
\end{equation}
Finally, to attain the optimal quantum success probability in the $3-$bit POM task, Alice must prepare and send the following ensemble of states
\begin{equation}
    \rho_{x^\delta_1x^\delta_2x^\delta_3} = \frac{1}{2}\qty(\openone_2 + \frac{(-1)^{x^\delta_1}\sigma_x+(-1)^{x^\delta_2}\sigma_y+(-1)^{x^\delta_3}\sigma_z}{\sqrt{3}})\otimes\frac{\openone_{J}}{J}
\end{equation}

This proves the existence of unitary which leads to a map from the physical experiment to the reference experiment.


\subsection{Construction of unitary for $5-$bit POM task}

The condition for achieving the optimal success probability in the $5-$bit POM task requires that the observables on Bob's side mutually anticommute, i.e., $\{B_y,B_{y'}\} = 0, \forall y \neq y'\in \{1,2,3,4,5\}$. This implies the existence of five mutually anticommuting observables, which in turn necessitates that Bob's Hilbert space must have dimension at least $2^m$, where $m=\lceil (n-1)/2 \rceil$. For $n=4$ or $5$, this gives $m=2$, resulting in minimum Hilbert space dimension of $d=4J$ for some $J\geq 1$. Accordingly, in the physical experiment, we consider an extended Hilbert space formed via a tensor product of $\mathcal{H}^4$ and $\mathcal{H}^J$ to support such observables. Without loss of generality, we assume that the first observable is diagonalised in some basis, via a suitable unitary $U_1: \mathcal{H}^d\to\mathcal{H}^4\otimes\mathcal{H}^J$ such that
\begin{equation} \label{b13}
    \mathtt{B}_1 = U_1B_1U_1^\dagger =\qty( \sigma_z\otimes\openone_2)\otimes\openone_{J}
\end{equation}
where $\sigma_z, \openone_2 $ are operators acting on $\mathcal{H}^2$. Using the anticommutation relations $\{B_1,B_y\}=0 \forall y \in \{2,3,4,5\}$ and applying $\mathtt{B}_y = U_1B_yU_1$, the remaining observables can be expressed in the same basis as 
\begin{equation}
    \begin{aligned}
    \mathtt{B}_2 &= \begin{bmatrix}
            0 & X_2\\
            X_2^\dagger &0
        \end{bmatrix} \ ; \ 
        \mathtt{B}_3 = \begin{bmatrix}
            0 & X_3\\
            X_3^\dagger &0
        \end{bmatrix} \ ; \\
        \mathtt{B}_4 &= \begin{bmatrix}
            0 & X_4\\
            X_4^\dagger &0
        \end{bmatrix} \ ; \
        \mathtt{B}_5 = \begin{bmatrix}
            0 & X_5\\
            X_5^\dagger &0
        \end{bmatrix} \ ;
    \end{aligned}
\end{equation}
where each $X_y$ is a $2J \times 2J$ matrix, with $J=d/4$. Now, consider the following unitary operator $U_2:\mathcal{H}^d \to \mathcal{H}^d$, defined by
\begin{equation}
    U_2 = \begin{bmatrix}
        \openone_{2J} & 0\\
        0 & \iota X_2
    \end{bmatrix}
\end{equation}
It is straightforward to verify that $U_2$ leaves $B_1$ invariant, \textit{i.e.}, $U_2\qty( \sigma_z\otimes\openone_2\otimes\openone_{J})U_2^\dagger = \sigma_z\otimes\openone_2\otimes\openone_{J}$. Applying $U_2$ to $\mathtt{B}_2$ yields
\begin{equation}\label{b23}
     \begin{aligned}
        U_2\mathtt{B}_2U_2^\dag &=\begin{bmatrix}
            0 & -\iota\openone_{J}\\
            \iota\openone_{J}^\dagger & 0
        \end{bmatrix} = \qty(\sigma_y\otimes\openone_2)\otimes\openone_{J}
    \end{aligned}
\end{equation}
From the anticommutation condition $\{\mathtt{B}_2,\mathtt{B}_y\} =0$ for $y\geq3$, it follows that $X_2 = -X_yX_2^\dagger X_y$. Applying $U_2$ to $B_y$ for $y\geq3$ then yields
\begin{equation}
        U_2\mathtt{B}_yU_2^\dag = \sigma_x\otimes\qty(-\iota X_yX_2^\dagger) \ \forall y \in \{3,4,5\}
\end{equation}
Note that we already obtain explicit forms for $B^\prime_1$ and $B^\prime_2$ in the known dimension $d'=4$, from Eqs.~(\ref{b13}) and (\ref{b23}). We now define a new set of Hermitian operators acting on $\mathcal{H}^{2J}$ as $\mathcal{B}_y := (-\iota X_y X_2^\dagger)$ for $y\geq 3$. From the mutual anticommutation $\{\mathtt{B}_y,\mathtt{B}_{y^\prime}\}=\delta_{y y'}\openone_{2J}$ implies $\{\mathcal{B}_y,\mathcal{B}_{y^\prime}\}=\delta_{y y'}\openone_{2J}$ for $y,y'\in\{3,4,5\}$. Using the standard construction for three mutually anticommuting observables (as previously discussed in Sec.~\ref{st3bit}), there exists a unitary $U_3:\mathcal{H}^{2J}\to \mathcal{H}^2\otimes\mathcal{H}^J$ such that 
\begin{equation}
\begin{aligned}
        U_3\mathcal{B}_3U_3^\dag&=U_3\qty(-iX_3X_2^\dag )U_3^\dag = \sigma_z\otimes\openone_{J}\\
        U_3\mathcal{B}_4U_3^\dag&=U_3\qty(-iX_4X_2^\dag )U_3^\dag = \sigma_y\otimes\openone_{J}\\
        U_3\mathcal{B}_5U_3^\dag&=U_3\qty(-iX_5X_2^\dag )U_3^\dag = \sigma_x\otimes\openone_{J}
\end{aligned}
\end{equation}
Thus, we define the overall unitary transformation as
\begin{equation}
  \exists U:\mathcal{H}^d\to \mathcal{H}^4\otimes\mathcal{H}^J \  \text{s.t.} \ U = \qty(\openone_2\otimes U_3)U_2U_1
\end{equation}
which maps the observables $\{B_y\}$ as follows
\begin{equation}
    \begin{aligned}
        UB_1U^\dag & = \qty(\sigma_z\otimes\openone_2)\otimes\openone_{J}; \ \  UB_2U^\dag  = \qty(\sigma_y\otimes\openone_2)\otimes\openone_{J}\\
        UB_3U^\dag & = \qty(\sigma_x\otimes\sigma_z)\otimes\openone_{J}; \ \ 
        UB_4U^\dag  = \qty(\sigma_x\otimes\sigma_y)\otimes\openone_{J}\\
        UB_5U^\dag &= \qty(\sigma_x\otimes\sigma_x)\otimes\openone_{J}\\
    \end{aligned}
\end{equation}
The state transforms accordingly as
\begin{equation*}
    U\qty(\rho_{x^\delta_1x^\delta_2x^\delta_3x^\delta_4x^\delta_5})U^\dag = \frac{1}{4}\qty(\openone_4 +\frac{ \sum_{y=1}^5(-1)^{x^\delta_y}B^\prime_y}{\sqrt{5}})\otimes\frac{\openone_{J}}{J}
\end{equation*}
This construction reveals a recursive structure in obtaining optimal observables through successive local unitary transformations. The case of the $4$-bit POM task arises as a special instance of this procedure. The generalisation to arbitrary $n-$bit tasks follows analogously and is presented in Appx.~\ref{appx2}.


\section{Summary and Discussion} \label{conclu}

In sum, we have demonstrated a fully DI self-testing protocol within a constrained prepare-measure scenario devoid of assuming the dimension of the quantum system, a feat which was believed to be unachievable. In order to showcase this, we considered a well-known prepare-measure communication game, known as the $n$-bit POM task \cite{Spekkens2009}, wherein a parity obliviousness constraint is imposed on the preparations. This constraint plays a crucial role in facilitating a dimension-independent derivation of the optimal quantum bound of the success probability which is greater than that of achievable by a preparation noncontextual ontological model. 

We explicitly proved DI self-testing of both the preparation and measurement devices for the cases of $n=2$ and $n=3-$bit POM tasks. Subsequently, we generalised our argument to arbitrary $n$. Our results show that achieving the optimal quantum value in the $n-$bit POM task warrants that Bob's measurements must correspond to $n$ mutually anticommuting observables. This requirement in turn imposes a lower bound on the Hilbert space dimension of the involved quantum system, constraining it to $d=2^m$, where $m=\lceil \frac{1}{2}(n-1)\rceil$. 

We further characterised the structure of Alice's optimal preparations, which correspond to the $2^n$ vertices of an $n-$dimensional regular hypercube embedded within a subset of the generalised Bloch sphere associated with $\mathcal{H}^d$. Geometrically, these states lie on the surface of a unit-radius $n-$ dimensional sphere, forming a structure we refer to as the Clifford sphere. The angular separation between any two preparations is determined by the Hamming distance between their respective bit-strings, with complementary bit-pairs mapped to antipodal points on the sphere, as illustrated in Fig.~\ref{fig:clifford-bloch}. 

Our self-testing protocol certifies both the quantum states and the measurements through an explicitly constructed local unitary, which maps any unknown finite-dimensional strategy onto a reference strategy realised in the minimal Hilbert space. This unitary preserves all relevant measurement statistics, thereby uniquely identifying the underlying quantum resources (up to local isometries and complex conjugation) solely on the basis of the observed success probability.

The fully DI nature of our scheme enables its direct application to randomness certification. Attaining the optimal quantum value in the POM task rules out the existence of any preparation-noncontextual ontological model, thereby guaranteeing that Bob's measurement outcomes exhibit intrinsic randomness. This observation paves the way for single-device randomness expansion protocols, in which certified randomness can be generated even in the presence of quantum side information potentially held by an adversary.

Finally, we highlight that the necessity of $n$ mutually anticommuting observables to achieve the optimal quantum value intrinsically fixes the minimal dimension of the underlying quantum system.  Hence, the POM task functions as an DI quantum dimension witness, offering an operational method for dimension certification. Our findings also motivate the development of DI self-testing protocols for other prepare-and-measure communication games, particularly those incorporating alternative operational constraints or symmetry conditions \cite{pan2019, Abhyoudai2023, Paul2024, Roy2024}.


\section*{Acknowledgements}

SS acknowledges fruitful discussions with Som Kanjilal regarding the structure of Clifford algebras, which significantly contributed to clarifying the geometric structures analysed in this work. RKS acknowledges the financial support from the Council of Scientific and Industrial Research (CSIR, 09/1001(17051)/2023-EMR-I), Government of India. SS acknowledges the support from the National Natural Science Fund of China (Grant No. G0512250610191) and the local hospitality from the research grant SG-160 of IIT Hyderabad, India. SN acknowledges the support from the research grant MTR/2021/000908, Government of India. AKP acknowledges the support from Research Grant No. SERB/CRG/2021/004258, Government of India.



\appendix
\onecolumngrid


\section{Derivation of classical optimal success probability for $2,3-$bit POM task}\label{copt23}

In $2-$bit POM, Alice randomly receives a $2-$bit string $x^{\delta}\in\{00,01,10,11\}$.  Given that the parity set comprises a single element, $\mathbb{P}_2=\{11\}$, we observe that for $x^{\delta}\in \{00,11\}$, the parity condition satisfies $x^\delta.s = 1$, whereas for $x^{\delta}\in\{01,10\}$, we have $x^\delta.s = 0$. This means that to ensure parity obliviousness, the following condition must uphold
\begin{equation}\label{po1}
\forall y \ \ \ p(b|\rho_{00},B_y) + p(b|\rho_{11},B_y)=p(b|\rho_{01},B_y) + p(b|\rho_{10},B_y)   
\end{equation}
This simply means that in an ontological model 
\begin{equation}
\label{po11}
 \forall \lambda \ \ \      \mu_{00}(\lambda)+\mu_{11}(\lambda)=\mu_{01}(\lambda)+\mu_{10}(\lambda)
\end{equation}
Given a $\lambda$, the response function is given by $\xi(b|E^{b}_k,\lambda)$ where $k=1,2$. Now, the success probability in the ontological model is given by
\begin{equation}
    \begin{aligned}
        (\mathcal{S}_2)_C &= \frac{1}{8}\int\bigg[\qty(\mu_{00}(\lambda)+\mu_{01}(\lambda))\xi(0|E^{0}_1,\lambda)+\qty(\mu_{00}(\lambda)+\mu_{10}(\lambda))\xi(0|E^{0}_2,\lambda)\\
        &+\qty(\mu_{11}(\lambda)+\mu_{10}(\lambda)\xi(1|E^{1}_1,\lambda))+\qty(\mu_{11}(\lambda)+\mu_{01}(\lambda))\xi(1|E^{1}_2,\lambda)\bigg]d\lambda
    \end{aligned}
\end{equation}
Using Eq. (\ref{po11}) and by noting $\xi(1|E^{1}_k,\lambda)=1-\xi(0|E^{0}_k,\lambda)$ we obtain

\begin{equation}
        \qty(\mathcal{S}_2)_C = \frac{1}{2} + \frac{1}{4}\bigg[\int\mu_{00}(\lambda)\qty(\xi(0|E^{0}_1,\lambda)+\xi(0|E^{0}_2,\lambda))d\lambda- \int\mu_{10}(\lambda)\xi(0|E^{0}_1,\lambda)d\lambda
        - \int\mu_{01}(\lambda)\xi(0|E^{0}_2,\lambda)d\lambda\bigg]
\end{equation}
For maximization of success probability, we need to consider that there exists at least a $\lambda$ for which $\xi(0|E^{0}_1,\lambda)=\xi(0|E^{0}_2,\lambda)=1$. Thus we have the following expression for success probability
 \begin{equation}
        \qty(\mathcal{S}_2)_C = \frac{1}{2} + \frac{1}{4}\qty[2\int\mu_{00}(\lambda)d\lambda- \int\mu_{10}(\lambda)d\lambda
        - \int\mu_{01}(\lambda)d\lambda]
\end{equation}
To maximize, we need to have non-zero support of $\lambda$ to $\mu_{00}(\lambda)$ and zero support to $\mu_{01}(\lambda)$ and $\mu_{10}(\lambda)$. But to satisfy Eq. (\ref{po11}), both $\mu_{01}(\lambda)$ and $\mu_{10}(\lambda)$ cannot have nonzero support. Therefore, we have the maximum success probability
\begin{equation}
    \qty(\mathcal{S}_2)_C = \frac{1}{2}\qty(1+\frac{1}{2})
\end{equation}

For the $3$-bit POM task, parity set is $\mathbb{P}_3=\{011,101,110,111\}$. The parity oblivious constraint dictates that for any $\lambda$ in an ontological model the following conditions need to be satisfied.
\begin{align}
    &s=011: \ \mu_{000}(\lambda) + \mu_{011}(\lambda) + \mu_{100}(\lambda) + \mu_{111}(\lambda) = \mu_{110}(\lambda) + \mu_{101}(\lambda) + \mu_{010}(\lambda) + \mu_{001}(\lambda) \label{pos31}\\
&s=101: \mu_{000}(\lambda) + \mu_{001}(\lambda) + \mu_{110}(\lambda) + \mu_{111}(\lambda) = \mu_{010}(\lambda) + \mu_{100}(\lambda) + \mu_{011}(\lambda) + \mu_{101}(\lambda) \label{pos32}
\\
&s=110: \ \mu_{000}(\lambda) + \mu_{010}(\lambda) + \mu_{101}(\lambda) + \mu_{111}(\lambda) = \mu_{001}(\lambda) + \mu_{100}(\lambda) + \mu_{011}(\lambda) + \mu_{110}(\lambda)\label{pos33}
\\
&s=111: \ \mu_{000}(\lambda) + \mu_{011}(\lambda) + \mu_{101}(\lambda) + \mu_{110}(\lambda) = \mu_{001}(\lambda) + \mu_{010}(\lambda) + \mu_{100}(\lambda) + \mu_{111}(\lambda) \label{pos34}  
\end{align}
Taking pairwise sum and difference of the above parity-oblivious conditions, we obtain the following set of equations
\begin{equation}\label{depo}
    \begin{aligned}
        \mu_{000}(\lambda) + \mu_{111}(\lambda) &= \mu_{101}(\lambda) + \mu_{010}(\lambda)  = \mu_{001}(\lambda) + \mu_{110}(\lambda)\\
        &= \mu_{010}(\lambda) + \mu_{001}(\lambda) = \mu_{100}(\lambda) + \mu_{011}(\lambda) \\
        &= \mu_{001}(\lambda) + \mu_{110}(\lambda) = \mu_{011}(\lambda) + \mu_{100}(\lambda)\\
        &\\
        \mu_{100}(\lambda) + \mu_{111}(\lambda) &= \mu_{101}(\lambda) + \mu_{110}(\lambda), \quad \mu_{000}(\lambda) + \mu_{011}(\lambda) = \mu_{010}(\lambda) + \mu_{001}(\lambda)\\
        \mu_{001}(\lambda) + \mu_{111}(\lambda) &= \mu_{011}(\lambda) + \mu_{101}(\lambda),\quad\mu_{000}(\lambda) + \mu_{110}(\lambda) = \mu_{010}(\lambda) + \mu_{100}(\lambda),\\
        \mu_{010}(\lambda) + \mu_{111}(\lambda) &= \mu_{011}(\lambda) + \mu_{110}(\lambda), \quad \mu_{000}(\lambda) + \mu_{101}(\lambda) = \mu_{001}(\lambda) + \mu_{100}(\lambda)
    \end{aligned}
\end{equation}
The expression for maximum success probability is 
\begin{equation}
\begin{aligned}
    (\mathcal{S}_3)_C = \frac{1}{24}
\Big[
&\int\qty(\mu_{000}(\lambda) + \mu_{001}(\lambda) + \mu_{010}(\lambda) + \mu_{011}(\lambda)) \xi(0 \mid E_1^0, \lambda)d\lambda\\
+ &\int\qty(\mu_{000}(\lambda) + \mu_{001}(\lambda) + \mu_{100}(\lambda) + \mu_{101}(\lambda)) \xi(0 \mid E_2^0, \lambda)d\lambda\\
+ &\int\qty(\mu_{000}(\lambda) - \mu_{010}(\lambda) + \mu_{100}(\lambda) + \mu_{110}(\lambda)) \xi(0 \mid E_3^0, \lambda)d\lambda\\
+ &\int\qty(\mu_{100}(\lambda) + \mu_{101}(\lambda) + \mu_{110}(\lambda) + \mu_{111}(\lambda)) \xi(1 \mid E_1^1, \lambda)d\lambda\\
+ &\int\qty(\mu_{010}(\lambda) + \mu_{011}(\lambda) + \mu_{110}(\lambda) + \mu_{111}(\lambda)) \xi(1 \mid E_2^1, \lambda)d\lambda\\
+ &\int\qty(\mu_{001}(\lambda) + \mu_{011}(\lambda) + \mu_{101}(\lambda) + \mu_{111}(\lambda)) \xi(1 \mid E_3^1, \lambda)d\lambda\Big]
\end{aligned}
\end{equation}

This upon simplification using the fact that $\xi(1|E^{1}_k,\lambda)=1-\xi(0|E^{0}_k,\lambda)$, and using relations in Eq.~(\ref{depo}) we obtain
\begin{equation}
\begin{aligned}
(\mathcal{S}_3)_C = \frac{1}{2}+\frac{1}{24}
\Bigg[
&\int\qty[\mu_{000}(\lambda) + \mu_{011}(\lambda)-\mu_{100}(\lambda) - \mu_{111}(\lambda)]\xi(0 \mid E_1^0, \lambda)d\lambda \\
+ &\int\qty[\mu_{000}(\lambda) + \mu_{101}(\lambda)-\mu_{010}(\lambda) - \mu_{111}(\lambda)] \xi(0 \mid E_2^0, \lambda)d\lambda \\
+ &\int\qty[\mu_{000}(\lambda) + \mu_{110}(\lambda) - \mu_{001}(\lambda)- \mu_{111}(\lambda)] \xi(0 \mid E_3^0, \lambda)d\lambda\Bigg]
\end{aligned}
\end{equation}
Similar to the argument for $2-$bit case, we assume that  there exists at least a $\lambda$ for which $\xi(0|E^{0}_1,\lambda)=\xi(0|E^{0}_2,\lambda)=\xi(0|E^{0}_3,\lambda)=1$. Thus we have the following expression for success probability 
\begin{equation}
(\mathcal{S}_3)_C= \frac{1}{2}+\frac{1}{24}
\Bigg[
\int\qty[3\mu_{000}(\lambda) + \mu_{011}(\lambda)+ \mu_{101}(\lambda) + \mu_{110}(\lambda)-\mu_{100}(\lambda)-\mu_{010}(\lambda)- \mu_{001}(\lambda) - 3\mu_{111}(\lambda)]d\lambda \Bigg]
\end{equation}
We maximize $(\mathcal{S}_3)_C$ by appropriately considering the support of $\lambda$ satisfying the constraint in Eq. (\ref{depo}). We finally derive
\begin{equation}
    \qty(\mathcal{S}_3)_C= \frac{1}{2} \qty(1+\frac{1}{ 3})
\end{equation}
which matches with success probability derived in \cite{Spekkens2009}. Following the derivation for $n=2$ and $n=3$, one can derive the success probability for an arbitrary $n$. We have checked upto $n=8$, which is in accordance with the Eq.~(\ref{clprob}).

\section{Derivation of equality of anticommutation of $\{A_i,A_j\}$ for $n=3$ (Eq.~\ref{acpo3})} \label{apnA}

The parity oblivious conditions are
\begin{align}
    &s=011: \ A_0 + A_3 + A_4 + A_7 = A_6 + A_5 + A_2 + A_1 \label{poobs31}\\
&s=101: \ A_0 + A_1 + A_6 + A_7 = A_2 + A_4 + A
_3+ A_5 \label{poobs32}
\\
&s=110: \ A_0 + A_2 + A_5 + A_7 = A_1 + A_4 + A_3 + A_6\label{poobs33}
\\
&s=111: \ A_0 + A_3 + A
_5+ A_6 = A_1 + A_2 + A_4 + A_7 \label{poobs34}  
\end{align}
Substrating Eq.~(\ref{poobs31}) and (\ref{poobs32}), we obtain
\begin{equation}
\begin{aligned}
    A_0 + A_3 + A_4 + A_7 - (A_0 + A_1 + A_6 + A_7) &= A_6 + A_5 + A_2 + A_1 - (A_2 + A_4 + A_3+ A_5) \\
    A_3 + A_4 &= A_1+A_6
 \end{aligned}   
\end{equation}
Similarly Substrating Eq.~(\ref{poobs31}) and (\ref{poobs33}) gives $A_3+A_4 = A_2+A_5$ and Eq.(\ref{poobs31})+(\ref{poobs32}) gives $A_0+A_7=A_2+A_5$. Hence we have $A_0+A_7 = A_1+A_6 = A_2+A_5 = A_3 + A_4$. Squaring the both sides of the resultant relations, we obtain $\{A_0,A_7\} = \{A_1,A_6\} = \{A_2,A_5\} = \{A_3 , A_4\}$.
\section{Derivation of equality of anticommutation of $\{A_i,A_j\}$ for $n-$bit POM task (Eq.~\ref{acpon})} \label{apnB}

The parity-obliviousness condition for every parity vector $s\in\mathbb{P}_n$ (i.e. of Hamming weight $\geq 2$), where the set $\mathbb{P}_n$ is defined in Eq.~(\ref{porac set}). Considering each state is of the form $\rho_{x^{\delta}}=\frac{1}{d}\qty(\openone+A_\delta)$, with $A_\delta$ being traceless, dichotomic Hermitian operators, the condition follows from Eq.~(\ref{postaten}) becomes
\begin{equation}\label{poobsnap}
    \forall s,B,b \ \sum\limits_{x^{\delta}|x^{\delta}.s=0}A_{x^{\delta}}=\sum\limits_{x^{\delta}|x^{\delta}.s=1}A_{x^{\delta}} \implies \sum_{x^{\delta}}(-1)^{x^{\delta}\cdot s} A_X^{\delta} =0
\end{equation}
This is a Fourier-type orthogonality condition: the sum over the operators, weighted by the parity character $(-1)^{x^{\delta}\cdot s}$, must vanish for all non-trivial $s$.

Let $x^{\overline{\delta}}$ denotes the bitwise compliment of $x^{\delta}\in\{0,1\}^n$, with $\overline{\delta}$ its decimal index. The following identity holds
\begin{equation}
    x^\delta \cdot s \oplus x^{\overline{\delta}}\cdot s=1 \implies (-1)^{x^\delta \cdot s}+(-1)^{x^{\overline{\delta}}\cdot s}=0
\end{equation}
Substituting this into Eq.~(\ref{poobsnap}), we obtain
\begin{equation} \label{poobsnap1}
    \sum_{x^{\delta}=0}^{2^{n-1}-1}(-1)^{x^{\delta}\cdot s} A_\delta =0 \implies \sum_{x^{\delta}=0}^{2^{n-1}-1}(-1)^{x^{\delta}\cdot s} \qty(A_\delta-A_{\overline{\delta}})=0
\end{equation}
Each parity condition $s\in\mathbb{P}_n$ defines a parity character $\chi_s(x^\delta):=(-1)^{x^{\delta} \cdot s}$. We may view $D_{\delta}:=\qty(A_\delta-A_{\overline{\delta}})$ as a function on $x^{\delta}\in{0,1}^n$. Then parity-obliviousness says that the Fourier coefficient of $D_\delta$ along any $s\in \mathbb{P}_n$ must be zero, implying $\braket{D_\delta,\chi_s}=\sum_\delta (-1)^{x^{\delta} \cdot s}D_\delta=0\implies D_\delta \perp \chi_s$. Therefore, the set of all $D_\delta$ lies in the subspace orthogonal to the one spanned by the parity characters $\chi_s$, i.e., the parity functions. This orthogonality imposes that the algebraic structure of each $D_{\delta}$ is constrained identically across $\delta$, implying their spectral properties—such as norms or eigenvalue distributions—are uniform (up to unitary equivalence).

Now consider the anticommutator
\begin{equation}
    \{A_\delta,A_{\overline{\delta}}\}=A_\delta A_{\overline{\delta}}+A_{\overline{\delta}}A_\delta =\qty(A_\delta+
    A_{\overline{\delta}})^2-2 \openone \ \ (\text{using} \ A_\delta^2=A_{\overline{\delta}}^2=\openone)
\end{equation}
This implies that the anticommutator is determined entirely by the sum $S_\delta:=(A_\delta+A_{\overline{\delta}})$. Since $A_\delta=\frac{1}{2}\qty(S_\delta+D_\delta)$, substituting into Eq.~(\ref{poobsnap}) gives
\begin{equation}
    \sum_{x^{\delta}=0}^{2^{n-1}-1}(-1)^{x^{\delta}\cdot s} A_\delta =0 \implies  \sum_{x^{\delta}=0}^{2^{n-1}-1}(-1)^{x^{\delta}\cdot s} \qty(S_\delta+D_\delta) =0 \implies \sum_{x^{\delta}=0}^{2^{n-1}-1}(-1)^{x^{\delta}\cdot s} S_\delta =0 \ \ (\text{since} \ D_\delta \perp \chi_s)
\end{equation}
Hence, the set $\{S_\delta\}$ lies in the same orthogonal subspace and is constrained identically across $\delta$, i.e., $S_\delta \perp \chi_s$. It follows that all $S_\delta$ have the same spectral properties, and therefore all anticommutators $\{A_\delta,A_{\overline{\delta}}\}$ are equal across different $\delta$. Thus,
\begin{equation}
   \{A_\delta,A_{\overline{\delta}}\}=\{A_\eta,A_{\overline{\eta}}\} \ \ \forall \delta \neq \eta.
\end{equation}

\section{Derivation of equality anticommutation of Bob's observables $\{B_y,B_{y'}\}=0$} \label{apnC}

The Operator Norm $\omega_\delta$ is given by
\begin{equation}
\omega_\delta = \norm{\qty(\sum_{y = 1}^n (-1)^{x^\delta_y} B_y)} \ \forall x^{\delta}\in {0,1}^n; \ y \in [n]
\end{equation}
Bob's observables $B_y$ satisfies, $\Tr[B_y]=0$ and $B_y^2=\openone$. $\omega_\delta$ is the maximum eigenvalue magnitude of the observable $\sum_{y} (-1)^{x^\delta_y} B_y$.

To derive the optimality condition, we have applied the Cauchy–Schwarz inequality 
\begin{equation}
    \sum_{\delta=0}^{2^n-1-1} \omega_\delta \leq \sqrt{2^{n-1}}\sqrt{\sum_{\delta=0}^{2^n-1-1}\omega_\delta^2}
\end{equation}
Equality is attained \textit{iff} $\omega_i=\omega_j \ \forall i \neq j$. Evaluating the squared operator norm explicitly, we obtain 
\begin{equation}
\begin{aligned}
   \omega_\delta^2=\norm{\sum_{y=1}^n \epsilon_y B_y}^2 &= \norm{n \openone + \sum_{y\neq z}\epsilon_y \epsilon_z B_y B_z} \ \ \text{where} \  \epsilon_y:=(-1)^{x_y^{\delta}} \in \{\pm 1\} \\
    &= \norm{n \openone + \sum_{y < z}\epsilon_y \epsilon_z B_y B_z + \sum_{z < y}\epsilon_y \epsilon_z B_y B_z} \\
    &= \norm{n \openone + \sum_{y < z}\epsilon_y \epsilon_z B_y B_z + \sum_{y < z}\epsilon_y \epsilon_z B_z B_y} \ \ (\text{Changing the index of the last term}) \\
    &= \norm{n \openone + \sum_{y < z}\epsilon_y \epsilon_z \qty(B_y B_z + B_z B_y)} = \norm{n \openone + \sum_{y < z}\epsilon_y \epsilon_z \{B_y, B_z\}}
    \end{aligned}
\end{equation}
Now, if all $\omega_\delta$ are equal, it should be independent on the bit string $x^{\delta}$. Thus, the cross term appearing in $\omega_\delta^2$ must disappear, hence $\{B_y,B_z\}=0$ as $\epsilon_y\neq 0$. 
\section{Optimal preparations are the vertices of the $n$-dimensional Hypercube}\label{Appxcbghs}

From the optimality condition derived in Sec.~\ref{optqnc}, the states prepared by Alice take the form
\begin{equation}\label{gstate}
    \rho_{x^\delta} = \frac{1}{d} \ \qty[\openone+\sum_{y = 1}^n\qty(\frac{(-1)^{x^\delta_y}}{\sqrt{n}}) \ B_y ]
\end{equation}
where $B_y\in\mathscr{L}(\mathcal{H}^d)$ are traceless Hermitian operators satisfying the Clifford algebra relations $\{B_y,B_{y'}\}=2 \delta_{yy'} \openone$ \cite{Wehner2010}. We now interpret these states in terms of an $n$-dimensional Bloch-like decomposition over the orthonormal basis $\{B_1,\dots,B_n\}$, where the Bloch vector $\Vec{r}_{x^{\delta}} \in \mathbb{R}^n$,is given by
\begin{equation}
    \Vec{r}_{x^{\delta}}=\frac{(-1)^{x^\delta_y}}{\sqrt{n}}
\end{equation}
Given that $x^{\delta} \in (0,1)^n$, there exists $2^n$ such Bloch vectors. Each component of the vector takes the value $\pm \frac{1}{\sqrt{n}}$, depending on the corresponding bit value in $x^{\delta}$. Thus, all such vectors lies within the set of coordinates $\qty(\pm \frac{1}{\sqrt{n}},\dots,\pm \frac{1}{\sqrt{n}})$, forming the
$2^n$ vertices of an $n$-dimensional hypercube scaled by $\frac{1}{\sqrt{n}}$.  The Euclidean norm of each vector is given by 
\begin{equation}
   \norm{\vec{r}_{x^{
\delta}}}=\sqrt{\sum_{y=1}^n\qty(\frac{1}{\sqrt{n}})^2} =1
\end{equation}
Hence, all vectors lie on the surface of a unit-radius sphere in $\mathbb{R}^n$. This implies that the Clifford-Bloch vectors form a regular hypercube inscribed within the unit sphere. Each vector represents a direction that indicates the state's deviation from the maximally mixed state, but constrained entirely within the Clifford subspace.

To evaluate the distance between any two such vectors, $\vec{r}_{x^{\delta}}$ and $\vec{r}_{\tilde{x}^{\delta}}$, we write
\begin{equation}
\begin{aligned}
        \vec{r}_{x^{\delta}}&=\frac{1}{\sqrt{n}}\qty((-1)^{x_1^{\delta}},\dots,(-1)^{x_n^{\delta}}) \ \forall x^{\delta}\in(0,1)^n \\
        \vec{r}_{\tilde{x}^{\delta}}&=\frac{1}{\sqrt{n}}\qty((-1)^{\tilde{x}_1^{\delta}},\dots,(-1)^{\tilde{x}_n^{\delta}}) \ \forall \tilde{x}^{\delta}\in(0,1)^n 
\end{aligned}
\end{equation}
Let $h =Ham(x,\tilde{x})$ denote the hamming distance between the two bit strings $x^{\delta},\tilde{x}^{\delta}\in \{0,1\}^n$, i.e., the number of indices $i$ for which $x_\iota\neq \tilde{x}_i$. Then the difference is given by
\begin{equation}
    \vec{r}_{x^{\delta}}-\vec{r}_{\tilde{x}^{\delta}}=\frac{1}{\sqrt{n}}\qty((-1)^{x_1^{\delta}}-(-1)^{\tilde{x}_1^{\delta}},\dots,(-1)^{x_n^{\delta}}-(-1)^{\tilde{x}_n^{\delta}})
\end{equation}
Each differing bit contributes $(\pm 2)^2=4$, so the squared Euclidean distance is
\begin{equation}
    \norm{\vec{r}_{x^{\delta}}-\vec{r}_{\tilde{x}^{\delta}}}=\frac{1}{n}\sum_{y=1}^n\qty((-1)^{x_i^{\delta}}-(-1)^{\tilde{x}_i^{\delta}})^2=\frac{4h}{n} 
\end{equation}
This confirms that the Clifford Bloch vectors constitute the vertices of a regular $n$-dimensional hypercube of edge length $\frac{2}{\sqrt{n}}$, centred at the origin and inscribed within the unit sphere in $\mathbb{R}^n$.

However, this Clifford subspace does not span the entire quantum state space unless $n=d^2-1$. In general, the full Bloch sphere of a $d$-dimensional system resides in $\mathbb{R}^{d^2-1}$, whereas the Clifford subspace has dimension $n\leq d^2-1$. Since the maximal number of mutually anticommuting traceless Hermitian operators is at most $2m+1$ for $d=2^m$, equality holds only for $d=2$ (where Pauli matrices span the entire Bloch space) and $d=3$ (where Gell-Mann matrices span the full space). In higher dimensions, the Clifford subspace represents only a proper subspace, and the corresponding Bloch vectors form a structured slice of the full state space.

The restriction to dichotomic observables with spectrum $\{\pm 1\}$, is essential for ensuring anti-commutation and preserving the Clifford structure. This constraint inherently limits the dimension of the subspace these observables can span. Nonetheless, this dichotomicity is crucial for realising the hypercube structure within the Bloch representation, and for achieving the optimal quantum bias in POM tasks.

\begin{figure*}[t]
    \centering
    \begin{subfigure}[b]{0.32\textwidth}
        \includegraphics[width=\linewidth]{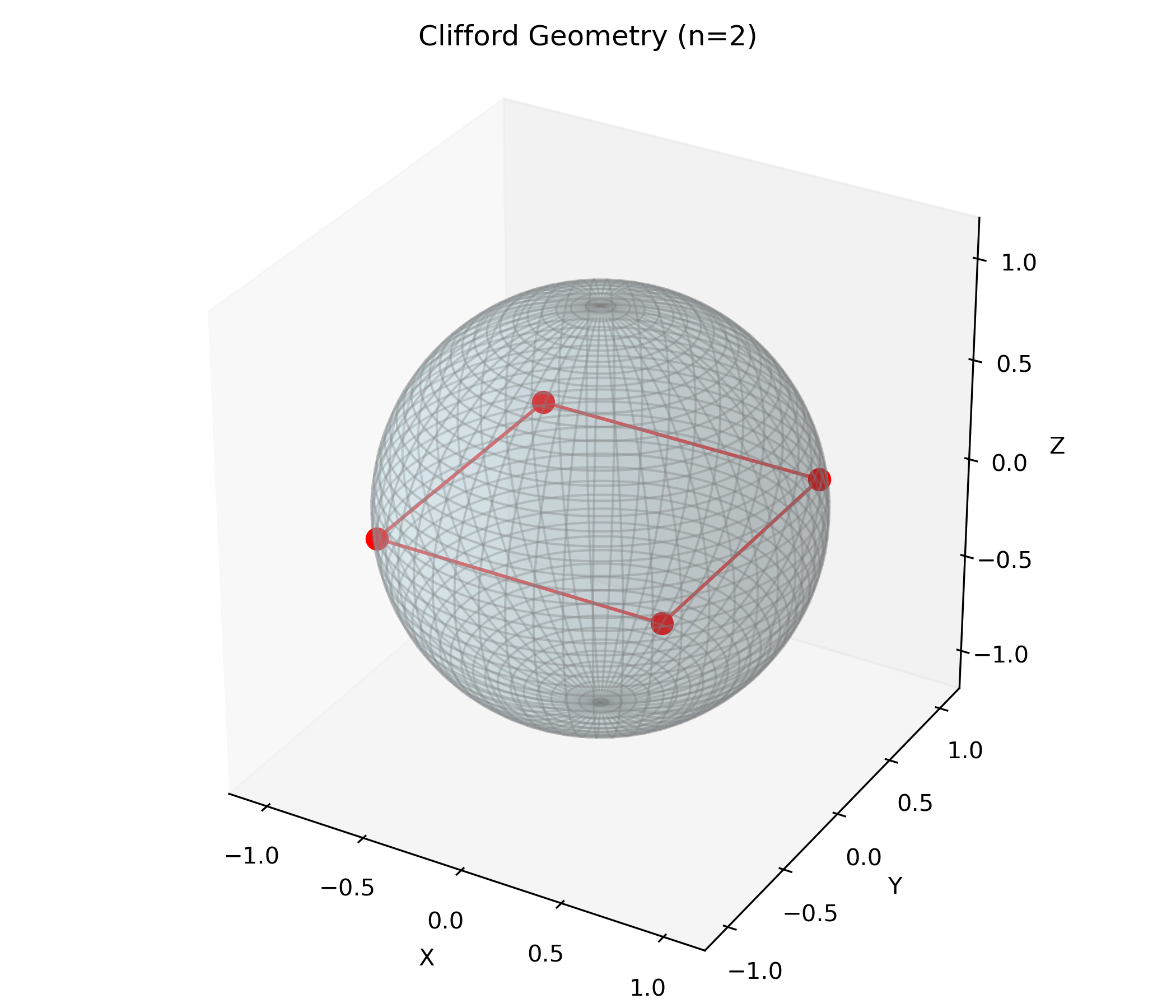}
        \caption{Bloch sphere with Clifford square ($n=2$)}
    \end{subfigure}
    \hfill
    \begin{subfigure}[b]{0.32\textwidth}
        \includegraphics[width=\linewidth]{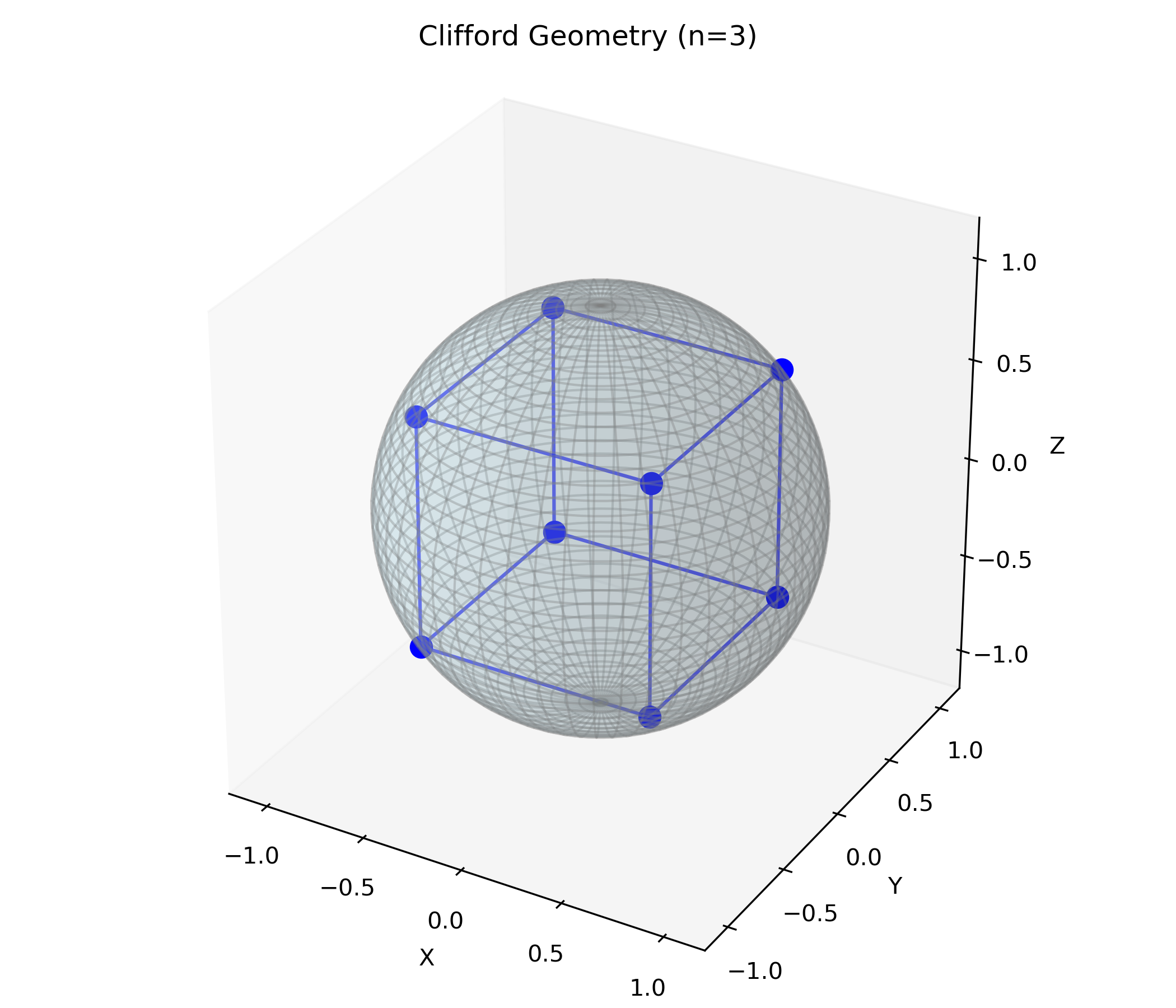}
        \caption{Clifford cube inside Bloch sphere ($n=3$)}
    \end{subfigure}
    \hfill
    \begin{subfigure}[b]{0.32\textwidth}
        \includegraphics[width=\linewidth]{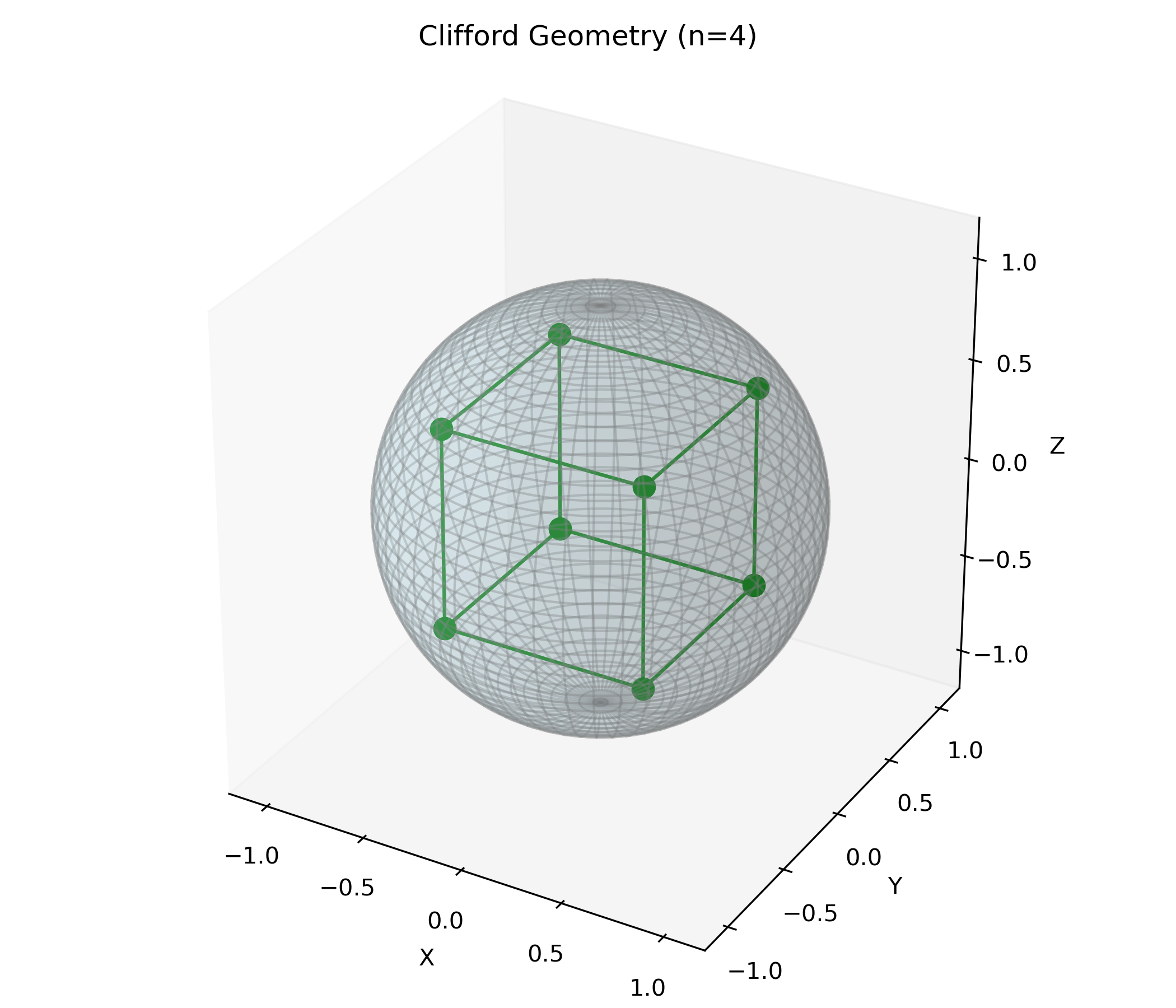}
        \caption{Projected 4D Clifford hypercube ($n=4$)}
    \end{subfigure}
    \caption{
        Geometric representation of Clifford-Bloch state preparation strategies for the $n$-bit POM task. 
        (a)~For $n=2$, the four prepared states form a square on the equatorial plane of the Bloch sphere, corresponding to the bit strings \texttt{00}, \texttt{01}, \texttt{10}, and \texttt{11}. Each vector points along a distinct Clifford direction generated by two mutually anticommuting observables. 
        (b)~For $n=3$, the eight preparation states form a cube inside the Bloch sphere, with each vertex associated with a 3-bit string. Clifford-Bloch vectors lie on the unit sphere and maintain equal pairwise distinguishability based on Hamming distance. 
        (c)~For $n=4$, the 16 Clifford-Bloch vectors define a 4D hypercube projected into 3D. Although only a subset of the structure is visible in projection, the geometry encodes all bit strings of length four, and the vertex connectivity reflects the Hamming-1 adjacency. All vectors lie on the unit sphere within the Clifford subspace and realise optimal preparations for the $n=4$ POM task.
    }
    \label{fig:clifford-bloch}
\end{figure*}


\section{Construction of unitary for $n$-bit POM task}\label{appx2}

The optimal success probability in the $n-$bit POM task imposes the condition that Bob’s observables $\qty{B_y \in \mathscr{L}\qty(\mathcal{H}^d)}_{y=1}^n$ mutually anticommute, i.e., 
\begin{equation}
    \{B_y,B_{y'}\}=2\delta_{yy'}\openone_d \ \forall y,y' \in \{1,2,\dots,n\}
\end{equation}
This condition implies that the set $\{B_y\}$ forms a representation of the complex Clifford algebra $\mathcal{C}l_n(\mathbb{C})$. The smallest dimension of a Hilbert space supporting an irreducible representation of this algebra is $d=2^m$, where $m=\lceil (n-1)/2\rceil$. Therefore, for a given $n$, the reference experiment requires $d^\prime=2^m$ to sufficiently support $n$ mutually anticommuting observables, \textit{i.e.} $\qty{B'_y\in\mathscr{L}(\mathcal{H}^{2^m})}$ for all $y\in[n]$. Thus, any realisation of such a set of observables on a Hilbert space of dimension $d$ can be decomposed as $\mathcal{H}^d=\mathcal{H}^{2^m}\otimes \mathcal{H}^J$, for some integer $J=2^m/d$, such that each observable $B_y$ is unitarily equivalent to (up to complex conjugation) $B'_y\otimes \openone_J$. Here, $B'_y$ belongs to a canonical set of Pauli-like matrices satisfying the required Clifford algebra. We now explicitly construct, via recursion, a unitary operator $\mathcal{U}\in \mathcal{U}(d)$ such that
\begin{equation}
 \exists \mathcal{U} : \mathcal{H}^d \to \qty(\Motimes_{i=1}^m \mathcal{H}^2 ) \otimes \mathcal{H}^J \ \text{s.t.} \ \mathcal{U} B_y \mathcal{U}^\dagger =B^\prime_y \otimes \openone_J \ \forall y \in \{1,2,\dots,n\}
\end{equation}
where $\{B^\prime_y\}$ generate $\mathcal{C}l_n (\mathbb{C})$ in the known dimension $d'=2^m$ (reference experiment).

\subsection{Initial step: Diagonalising the first observable $B_1$}\label{st1} We begin by constructing a unitary operator $U_1:\mathcal{H}^d\to\mathcal{H}^{2}\otimes\mathcal{H}^{2^{m-1}J}$ such that the first observable $B_1$ is mapped to the Pauli like observable $\qty(\sigma_z\Motimes_{i=2}^{m-1}\openone_2)$ acting on a $\qty(\Motimes_{i=1}^{m}\mathcal{H}^2)$ subspace
\begin{equation}\label{b1n}
 \exists U_1 : \mathcal{H}^d \to \qty(\mathcal{H}^{2}\otimes\mathcal{H}^{2^{m-1}})\otimes \mathcal{H}^J \ \text{s.t.} \ B_1 \mapsto \mathtt{B}_1:= U_1 B_1 U_1^{\dagger}=\qty(\sigma_z\Motimes_{i=2}^{m-1}\openone_2)\otimes\openone_{J} 
\end{equation}
for some $J = 2^m/d$. Note that $\{B_1,B_{y}\}=2\delta_{1,y}\openone_{d}\implies \{\mathtt{B}_1,\mathtt{B}_y\}=2 \delta_{1,y}\openone_{d} \ \forall y$. Consequently, in the basis where $\mathtt{B}_1$ is diagonalised via the unitary transformation $U_1$, each observable become off-diagonal with respect to this basis. Specifically, each transformed $\mathtt{B}_y:=U_1B_yU_1^\dagger$ takes the form
\begin{equation}\label{e5}
 B_y \mapsto \mathtt{B}_y:=U_1 B_y U_1^\dagger =\begin{bmatrix}
        0 & X_y\\
        X_y^\dagger & 0
    \end{bmatrix}, \ \forall y \in \{2,3,\dots,n\} \ \  \text{with} \ X_y \in \mathscr{L}(\mathcal{H}^{2^{m-1}J})
\end{equation}

\subsection{Diagonalising the second observable $B_2$}\label{st2} Now, choose one such off-diagonal observable, say $B_2$, and define the following unitary $U_2:\mathcal{H}^d\to \qty(\mathcal{H}^{2}\otimes\mathcal{H}^{2^{m-1}})\otimes \mathcal{H}^J$ 
\begin{equation}\label{e6}
    U_2=\begin{bmatrix}
        \openone_{2^{m-1}} & 0\\
        0 & \iota X_2
    \end{bmatrix}
\end{equation}
It is straightforward to verify that $U_2 \mathtt{B}_1 U_2^\dagger = \mathtt{B}_1$ and transforms $\mathtt{B}_2$ into
\begin{equation}\label{b2n}
    U_2 \mathtt{B}_2 U_2^\dagger = \qty(\sigma_y \Motimes_{i=2}^{m-1}\openone_2)\otimes \openone_J
\end{equation}

\subsection{Recursive step: Diagonalising observables $B_y$ for $y\geq 3$}\label{st3} For $y \geq 3$, the anticommutation $\{B_2,B_y\}=0 \implies X_y=-X_2X_y^\dagger X_2$. Now the action of $U_2$ on $\mathtt{B}_y$ is given by 
\begin{equation}
    U_2\mathtt{B}_yU_2 = \begin{bmatrix}
        0 & -\iota X_y X_2^\dag \\
         -\iota X_y X_2^\dag & 0
    \end{bmatrix}= \sigma_x\otimes\qty(-\iota X_yX_2^\dag) \ \ \forall y \in \{3,4,\dots,n\}
\end{equation}
Note that we already obtain explicit forms for $B^\prime_1$ and $B^\prime_2$ in some known dimension $d'=2^m$, from Eqs.~(\ref{b1n}) and (\ref{b2n}). Now to find the set of remaining $(n-2)$ mutually anticommuting observables $B_y^\prime$, we now define a new set of Hermitian operators acting on $\mathcal{H}^{2^{m-1}J}$ as $\mathcal{B}_y^{(1)} := (-\iota X_y X_2^\dagger)$ for $y\geq 3$. The anticommutation relation $\{\mathtt{B}_y,\mathtt{B}_{y^\prime}\}=2\delta_{y y'}\openone_{2^{m}J}$ implies $\qty{\mathcal{B}_y^{(1)},\mathcal{B}_{y^\prime}^{(1)}}=2\delta_{y y'}\openone_{2^{m-1}J}$ for $y,y'\in \{3,\dots,n\}$. Thus, such a set $\{\mathcal{B}^{(1)}_y\}_{y=3}^n$ forms a representation of the complex Clifford algebra $\mathcal{C}l_{n-2}(\mathbb{C})$ in $(n-2)$ dimensional subspace. The smallest dimension of a Hilbert space supporting an irreducible representation of this algebra is $d=2^{m-2}$, where $m=\lceil (n-1)/2\rceil$.

Following steps \ref{st1} and \ref{st2}, by suitably defining the unitary operators $V^{(1)}_1, V^{(1)}_2:\mathcal{H}^{2^{m-1}J}\to \mathcal{H}^{2^m-1}\otimes\mathcal{H}^J$ such that 
\begin{equation}\label{b33n}
    \begin{aligned}
      \qty(V^{(1)}_2  V^{(1)}_1) \mathcal{B}^{(1)}_3 \qty(V^{(1)}_2V^{(1)}_1)^\dagger&=\qty(\sigma_z \otimes \openone_{2^{m-2}})\otimes \openone_J \ ; \ \ \qty(V^{(1)}_2  V^{(1)}_1) \mathcal{B}^{(1)}_4 \qty(V^{(1)}_2V^{(1)}_1)^\dagger =\qty(\sigma_y \otimes \openone_{2^{m-2}})\otimes \openone_J \\
        \qty(V^{(1)}_2  V^{(1)}_1) \mathcal{B}^{(1)}_y \qty(V^{(1)}_2V^{(1)}_1)^\dagger&=\sigma_x \otimes \qty(-\iota \tilde{X}_y \tilde{X}_4^\dagger) \ \ \forall y \in \{5,6,\dots,n\}\\
    \end{aligned}
\end{equation}
$V^{(1)}_1, V^{(1)}_2,\tilde{X}_y \in \mathscr{L}(\mathcal{H}^{2^{m-1}J})$ are constructed in the similar spirit of Eqs.~(\ref{e5})and (\ref{e6}). Now, up to this point, we have a unitary $U'=\qty[\openone_2\otimes \qty(V^{(1)}_2  V^{(1)}_1)]U_2U_1$. It is straightforward to obtain the following.
\begin{equation}
\begin{aligned}
        U'B_1U'^\dag &= \qty(\sigma_z\Motimes_{i=2}^{m-1}\openone_2)\otimes\openone_{J} \\
        U'B_2U'^\dag &= \qty(\sigma_y\Motimes_{i=2}^{m-1}\openone_2)\otimes\openone_{J} \\
        U'B_3U'^\dag&=\qty(\sigma_x \otimes\sigma_z \otimes \openone_{2^{m-2}})\otimes \openone_J\\
        U'B_4U'^\dag&=\qty(\sigma_x \otimes\sigma_y \otimes \openone_{2^{m-2}})\otimes \openone_J\\
        U'B_yU'^\dag&=\sigma_x \otimes \qty(-\iota \tilde{X}_y \tilde{X}_4^\dagger) \ \ \forall y \in \{5,6,\dots,n\}
\end{aligned}
\end{equation}

Assume after $k$ steps, observables $\{B'_1,B'_2,\dots,B'_{2k}\}$ are in Pauli form on $\mathcal{H}^{2^k}$, while the remaining observables $\{B'_{2k+1},\dots,B'_{n}\}$ have the form $\qty(\Motimes_{i=1}^{k}\sigma_x )\otimes\mathcal{B}^{(k)}_y$ acting on $\qty(\Motimes_{i=1}^{k}\mathcal{H}^{2})\otimes \mathcal{H}^{(2^{m}-2k)J}$. Now, for $y=2k+1$, applying the unitary $V^{(k+1)}_2V^{(k+1)}_1$
\begin{equation}
\begin{aligned}
        V^{(k+1)}_2V^{(k+1)}_1 \mathcal{B}^{(k)}_{2k+1} \qty(V^{(k+1)}_2V^{(k+1)}_1)^\dagger&=\qty(\sigma_z\Motimes_{2k+1}^{m} \openone_{2})\otimes \openone_J \\  
    V^{(k+1)}_2V^{(k+1)}_1 \mathcal{B}^{(k)}_{2k+2} \qty(V^{(k+1)}_2V^{(k+1)}_1)^\dagger&=\qty(\sigma_y\Motimes_{2k+1}^{m} \openone_{2})\otimes \openone_J
\end{aligned}
\end{equation}
Remaining observables transform as $\qty(V^{(k+1)}_2V^{(k+1)}_1) \mathcal{B}^{(k)}_{y} \qty(V^{(k+1)}_2V^{(k+1)}_1)^\dagger\otimes \mathcal{B}^{k+1}_y$ with
\begin{equation}
    \qty{\mathcal{B}^{k+1}_y,\mathcal{B}^{k+1}_{y'}}=2\delta_{yy'}\openone_{2^{m-k-1}}
\end{equation}
Thus, such a set $\{\mathcal{B}^{(k+1)}_y\}_{y=2k+1}^n$ forms a representation of the complex Clifford algebra $\mathcal{C}l_{n-2k-1}(\mathbb{C})$ in $(n-2k-1)$ dimensional subspace, where $n$ is related to $m=\lceil (n-1)/2\rceil$.

The recursion terminates after $m=\lceil (n-1)/2\rceil$ steps.  The overall unitary $\mathcal{U}$ is then given by
\begin{equation}
    \mathcal{U}= \qty(\Motimes_{i=1}^{m}\openone_2 \otimes V^{(m)}_2V^{(m)}_1)\qty(\Motimes_{i=1}^{m-1}\openone_2 \otimes V^{(m-1)}_2V^{(m-1)}_1)\dots\qty(\Motimes_{i=1}^{k}  \openone_2 V^{(k)}_2V^{(k)}_1)\dots U_2U_1
\end{equation}
Then following the recursive steps, the action of $\mathcal{U}$ on Bob's observables $B_{y} \ y\in[n]$ is given by
\begin{equation}\label{fau}
\begin{aligned}
        \mathcal{U}B_1\mathcal{U}^\dag &= \qty(\sigma_z\Motimes_{i=2}^{m-1}\openone_2)\otimes\openone_{J} \\
        \mathcal{U}B_2\mathcal{U}^\dag &= \qty(\sigma_y\Motimes_{i=2}^{m-1}\openone_2)\otimes\openone_{J} \\
        \mathcal{U}B_3\mathcal{U}^\dag&=\qty(\sigma_x \otimes\sigma_z \otimes \openone_{2^{m-2}})\otimes \openone_J\\
        \mathcal{U}B_4\mathcal{U}^\dag&=\qty(\sigma_x \otimes\sigma_y \otimes \openone_{2^{m-2}})\otimes \openone_J\\
        \mathcal{U}B_5\mathcal{U}^\dag&=\qty(\sigma_x \otimes\sigma_x \otimes\sigma_z \otimes\openone_{2^{m-3}})\otimes \openone_J\\
        &\vdots \hspace{2cm} \vdots\\
       &\vdots \hspace{2cm} \vdots\\
        \mathcal{U}B_{n-1}\mathcal{U}^\dag&= \begin{cases}\qty(\Motimes_{i=1}^{m-1}\sigma_x \otimes\sigma_y)\otimes \openone_J & \text{if $n$ is odd} \vspace{0.3 cm} \\
        \qty(\Motimes_{i=1}^{m-1}\sigma_x \otimes\sigma_z)\otimes \openone_J & \text{if $n$ is even} \end{cases}\\
        &\\
        \mathcal{U}B_{n}\mathcal{U}^\dag&=\begin{cases}\qty(\Motimes_{i=1}^{m-1}\sigma_x \otimes\sigma_x)\otimes \openone_J & \text{if $n$ is odd} \vspace{0.3 cm} \\
        \qty(\Motimes_{i=1}^{m-1}\sigma_x \otimes\sigma_y)\otimes \openone_J & \text{if $n$ is even} \end{cases}
\end{aligned}
\end{equation}
Now, this Eq.~(\ref{fau}) can be expressed through more compact form, by introducing the new notation $B_y\mapsto B_{y,n}$ for any $n-$bit POM task, such that
\begin{eqnarray}
    \mathcal{U}B_{y,n}\mathcal{U}^\dag=
    \begin{cases}
    \qty(B_{y,n}^\prime \otimes \openone_{2^{m-1}})\otimes\openone_J  & \forall y\in \{1,2\} \vspace{0.3 cm}\\
\sigma_{x} \otimes B^\prime_{{y-2},{n-3}} & \forall y \in \{3,4,\dots,n\} 
    \end{cases}
\end{eqnarray}
where $B^\prime_{0,k} = \openone_2, B^\prime_{1,k} = \sigma_z$ and $B^\prime_{2,k} = \sigma_y \  \forall k \in [n-3]$. 


\twocolumngrid
\bibliography{references}


\end{document}